\documentclass[aps,prb,notitlepage,nofootinbib,floatfix]{revtex4-2}
\usepackage{graphicx}
\usepackage[dvipsnames]{xcolor}
\usepackage{amsfonts}
\usepackage{amssymb}
\usepackage{amsmath}
\usepackage{mathtools}
\usepackage{afterpage}
\usepackage[mathscr]{euscript}
\usepackage{braket}
\usepackage{subfigure}
\usepackage{multirow}
\usepackage{float}
\usepackage{qcircuit}
\usepackage{caption}
\usepackage{enumitem}
\usepackage{color}   
\usepackage{hyperref}
\hypersetup{
    colorlinks=true, 
    linktoc=all,     
    linkcolor=blue,  
}

\usepackage[utf8]{inputenc} 
\usepackage{yfonts}
\usepackage{dsfont}
\usepackage{cancel}
\usepackage[scr=boondox]{mathalpha}

\usepackage[pdf]{pstricks}
\usepackage{leftidx}
\usepackage{pst-all}
\usepackage{booktabs}

\usepackage{slashed}
\usepackage{stmaryrd}

\usepackage[bottom]{footmisc} 

\usepackage{amsthm,mathtools,amsmath}
\usepackage{tikz}

\newcommand{\Z}{\mathbb{Z}}
\newcommand{\R}{\mathbb{R}}

\newcommand{\C}{\mathbb{C}}

\newcommand{\tr}{\mathrm{tr}}

\usepackage{environ}
\NewEnviron{eqs}{%
\begin{equation}\begin{split}
    \BODY
\end{split}\end{equation}
}

\usepackage{hyperref}
\hypersetup{
    colorlinks=true,
    linkcolor=blue,
    filecolor=magenta,
    urlcolor=cyan,
}

\newtheorem{thm}{Theorem}
\newtheorem{lem}[thm]{Lemma}
\newtheorem{prop}[thm]{Proposition}
\newtheorem{cor}[thm]{Corollary}

\newcommand\restr[2]{{
  \left.\kern-\nulldelimiterspace 
  #1 
  \vphantom{\big|} 
  \right|_{#2} 
  }}

\usepackage{MnSymbol}
 
\usepackage{scalerel,amssymb}

\def\SO{\mathrm{SO}}
\def\Spin{\mathrm{Spin}}
\def\GL{\mathrm{GL}}
\def\SL{\mathrm{SL}}

\begin{document}

\title{Rank-$N$ Dimer Models on Surfaces}
\author{Sri Tata$^1$}
\affiliation{$^1$Yale University, Department of Mathematics}

\begin{abstract}
The web trace theorem of Douglas, Kenyon, Shi~\cite{dksWebs2022} expands the twisted Kasteleyn determinant in terms of traces of webs.
We generalize this theorem to higher genus surfaces and expand the twisted Kasteleyn matrices corresponding to spin structures on the surface, analogously to the rank-1 case of~\cite{cimasoniReshetikhin2007spin}. In the process of the proof, we give an alternate geometric derivation of the planar web trace theorem, relying on the spin geometry of embedded loops and a `racetrack construction' used to immerse loops in the blowup graph on the surface.
\end{abstract}

\maketitle

\twocolumngrid
\tableofcontents
\onecolumngrid

\section{Introduction}

In~\cite{dksWebs2022}, Douglas, Kenyon, Shi recently studied the planar $N$-fold dimer model coupled to $\SL_N$ graph connections. In particular, they considered twisting the Kasteleyn matrix $K$ of the graph by a graph connection $\Phi$ to produce a matrix $K(\Phi)$. Their main result is an expansion of the Kasteleyn determinant in terms of so-called `traces of multiwebs', generalizing an earlier $\SL_2$ result of Kenyon~\cite{kenyon2014doubledimer}. Using a combinatorial analysis of the Kasteleyn signs, their result reads (given a so-called `positive ciliation')
\begin{equation}
    \det(K(\Phi)) = \pm \sum_{\mathrm{multiweb}} \tr(\mathrm{multiweb},\Phi)
\end{equation}
for a global overall $\pm$ sign. 

In this paper, we give an alternate geometric derivation of this result using a mapping of Kastleyn theory to discrete spin structures on surfaces, following~\cite{cimasoniReshetikhin2007spin}. This allows us to immediately generalize the result of~\cite{dksWebs2022} to surfaces. In this setting, a geometric spin structure $\eta$ depends on the combinatorial data of a choice of Kasteleyn signs and a reference matching on the surface and gives a twisted Kasteleyn matrix $K_\eta(\Phi)$. Given $\eta$, we will make a natural definition of a web trace $\tr_\eta$ with respect to $\eta$ and show that
\begin{equation}
    \det(K_\eta(\Phi)) = \pm \sum_{\mathrm{multiweb}} \tr_\eta(\mathrm{multiweb},\Phi).
\end{equation}
Moreover, we show that on a genus $g$ surface
\begin{equation}
    \sum_{\mathrm{multiweb}} \tr(\mathrm{multiweb},\Phi)
    =
    \pm \frac{1}{2^g} \sum_{\substack{\eta \, \text{spin} \\ \text{structure}}}
    \mathrm{Arf}(\eta) \det(K_\eta(\Phi)) 
\end{equation}
where $\mathrm{Arf}(\eta) \in \{\pm 1\}$ is the so-called `Arf Invariant' of the spin structure. The above formula is in complete analogy to the formula~\cite{cimasoniReshetikhin2007spin} that given a bipartite graph on a surface with Kasteleyn matrices $K_\eta$ corresponding to spin structures $\eta$ that
\begin{equation}
    \#\{\text{dimer covers of graph}\}
    = 
    \pm \frac{1}{2^g} \sum_{\substack{\eta \, \text{spin} \\ \text{structure}}}
    \mathrm{Arf}(\eta) \det(K_\eta). 
\end{equation}
In turn, this formula is in complete analogy to the `higher-genus bosonization' formulas~\cite{alvarezGaume1987bosonHigherGenus} of conformal field theory which gives the same mapping between partition functions on higher-genus surfaces of of the dual theories, the Dirac fermion and the free boson at radius 1.

In this paper, we deal with much of the permutation sign combinatorics using the language of Grassmann variables. While the combinatorics could be handled without them, we prefer to use them for organizational purposes and to highlight the analogy to constructions in spin-TQFT which relate Grassmann calculus to combinatorial spin structures (c.f.~\cite{guWen2014Supercohomology,gaiottoKapustin2016,tata2020}). 

The organization of this paper is as follows. In Sec.~\ref{sec:preliminaries}, we review preliminary notions needed in the problem formulation: first the notions of webs, connections, and traces, and next how the notions of spin geometry arise from Kasteleyn theory. In Sec.~\ref{sec:rank_1_dimer}, we review the results of of~\cite{cimasoniReshetikhin2007spin} of rank-1 dimer models (i.e. single dimer covers) in our notations. Then in Sec.~\ref{sec:rank_N_dimers}, we formulate and prove our main results for rank-$N$ dimers. We conclude in Sec.~\ref{sec:discussion} with discussion and questions regarding this work.
In the Appendix, we review Grassmann variables in Appendix~\ref{app:grassmann}.
Relevant details regarding definitions and notions related to spin geometry are in Appendix~\ref{app:spin_structures} and can be read as a self-contained summary of 2D spin geometry.

\section{Preliminaries} \label{sec:preliminaries}
\subsection{(Multi-)webs, Connections, Traces, Blowup Graph} \label{sec:webs_and_traces}
Here, we review the main definitions of the blowup graph, multi-webs, and web traces that appear in~\cite{dksWebs2022} and throughout this work.

An $N$-web on a surface is a bipartite graph embedded in the surface for which all vertices have degree $N$. In general, a web may have multiple edges connecting two vertices $b,w$. As such, we can also refer to a `multiweb' as the data of all pairs of vertices $e = (b,w)$ connected to each other with a multiplicity $m_e$, such that the sums of multiplicities at each vertex $v$ is $N$, i.e. $\sum_{e \ni v} m_e = N$. Given a bipartite graph without multiple edges embedded in a surface, one can view a multiweb as a subgraph of the bipartite graph together with the multiplicity data on each edge.

A $\GL_N$ connection $\Phi$ on a bipartite graph is a collection of matrices $\{\phi_{b,w} \in \GL_N \,|\, (b \to w) \text{ a black-to-white edge}\}$. We also define $\phi_{w,b} = \phi_{b,w}^{-1}$ so that we can assign monodromy $\mathrm{mon}(L) = \phi_{v_1,v_2} \phi_{v_2,v_3} \cdots \phi_{v_n,v_1}$ to a loop $L = v_1 \to \cdots \to v_n \to v_1$.

We will soon review the definition of the trace of an abstract web, then we will show how the trace of a multiweb on the graph can be defined. 
Before defining the web trace, we need the notion of a `ciliation', which is an assignment of an ordering of edges around each vertex $v$ compatible with the cyclic ordering implied by the web embedding.
One can construct this as follows by drawing a tiny arrow, or cilium, at a vertex going into a face. For a black vertex $b$, we will label the incident edges $e_1(b),\cdots,e_{N}(b)$ in the \textit{clockwise order} with respect to the cilium. For a white $w$, the edges $e_1(w),\cdots,e_{N}(w)$ should be labeled in \textit{counterclockwise} order with respect to the cilium. Although for abstract webs, one can work with a general ciliation, we will be choosing to work with ciliations on multiwebs coming from a reference dimer cover $\mathcal{R}$ of the graph. 
\footnote{By a note in Section 3.6.1 of~\cite{dksWebs2022}, any $N$-regular bipartite graph has at least one dimer covering, so we can choose such a ciliation in general.}
In particular, if $e=(b,w)$ is in the matching $\mathcal{R}$, one can point the cilia of both $b,w$ into the face directly on the right of the edge $b \to w$ of the black-to-white direction. See Fig.~\ref{fig:blowup_graph_and_match} for a depiction of this case. 

Now we define the web trace with respect to a connection on the web $\{\phi_e \in \GL_N | e \, \text{edges}\}$. \footnote{Note that for this general definition, we can choose $\phi_e$, $\phi_{e'}$ to be different when $e,e'$ are distinct edges connecting the same vertices $b,w$. However in the setup of a multiweb embedded in a graph, like in this paper, we will have $\phi_{b,w}$ correspond to the single connection matrix between $b,w$.} We assign a copy of the vector space $V_e \cong \C^N$ and its dual $V_e^{*}$ for each edge $e$. The space $V_e$ will be associated to the black vertex of $e$ and the dual $V_e^{*}$ will be associated to the white one $w$. Denote by $\ket{1,e},\cdots,\ket{N,e}$ a chosen basis of each $V_e$ and $\bra{1,e},\cdots,\bra{N,e}$ a chosen basis of $V_e^*$.
For each black vertex $b$ and white vertex $w$, associate the tensors
\begin{equation}
\begin{split}
    \ket{b} &= \sum_{\sigma \in \mathrm{Sym}(N)} (-1)^{\sigma} \ket{\sigma(1),e_{1}(b)} \otimes \cdots \otimes \ket{\sigma(N),e_{N}(b)}
    \\
    \bra{w} &= \sum_{\sigma \in \mathrm{Sym}(N)} (-1)^{\sigma} \bra{\sigma(1),e_{1}(w)} \otimes \cdots \otimes \bra{\sigma(N),e_{N}(w)}.
\end{split}
\end{equation}
For each edge $e$, the matrices $\phi_{e}$ act om $V_e$ as 
\begin{equation}
    \phi_{e}(\sum_i v_i \ket{i,e}) = \sum_{i,j} (\phi_{e})_{i,j} v_j \ket{i,e}.
\end{equation}
From the tensors $\{\ket{b}, \bra{w}\}$ and the action of the $\{\phi_e\}$ we can form tensors
\begin{equation} \label{eq:web_tensor_defs}
    \bigotimes_{b} \ket{b} \in \bigotimes_e V_e
    \quad\quad\text{and}\quad\quad
    \left(\bigotimes_e \phi_e \right)
    \left(\bigotimes_b \ket{b} \right) \in \bigotimes_e V_e
    \quad\quad\text{and}\quad\quad
    \bigotimes_{w} \bra{w} \in \bigotimes_e V_e^{*}.
\end{equation}
Now, we define the web trace as:
\begin{equation}
    \tr(\mathrm{web},\Phi) := 
    \left( \bigotimes_{w} \bra{w} \right)
    \left( \bigotimes_e \phi_e    \right)
    \left( \bigotimes_{b} \ket{b} \right).
\end{equation}
At this point, we note that all of the web trace definitions implicitly depended on the ciliation. However, the ciliation enters only in the definitions of each $\ket{b}$ and $\bra{w}$. Changing the ciliation by a `click' (i.e. moving it from one face to an adjacent face) will reorder the permutations $\sigma$ in Eq.~\eqref{eq:web_tensor_defs} by an $N$-cycle, so changes the sign of $\ket{b}$ or $\ket{w}$ by $(-1)^{N+1}$. In particular, the ciliation only affects the web trace definitions up to an overall sign, and are completely invariant for $N$ odd.

First, we note that the web trace is multiplicative under connected components. In particular, if as a graph the web is a disjoint union $\mathrm{web} = \mathrm{web}_1 \sqcup \cdots \sqcup \mathrm{web}_k$ of webs, then it follows immediately that
\begin{equation}
    \tr(\mathrm{web},\Phi)
    =
    \prod_{\substack{\text{disconnected} \\ \text{components} \\ \mathrm{web}_i}} 
    \tr(\mathrm{web}_i,\Phi).
\end{equation}

Also note that all of these formulas implicitly depend on the choice of bases used on the edges. The above web trace formula has an associated gauge transformation formula associated to vertices $b$ or $w$. In particular, for a matrix $g \in \GL_N$ one can redefine the bases as $\widehat{\ket{j , e_k(b)}} = \sum_{i} g_{i,j} \ket{i , e_k(b)}$ simultaneously each edge $e_k(b)$ adjacent to $b$, or analogously for $w$. These bases are associated to tensors $\widehat{\ket{b}}$ and $\widehat{\bra{w}}$. One can see immediately from the definitions Eq.~\eqref{eq:web_tensor_defs} that this would give 
\begin{equation*}
    \widehat{\ket{b}} = \det(g) \ket{b}
    \quad\quad\text{or}\quad\quad
    \widehat{\bra{w}} = \det(g) \bra{w}.
\end{equation*}
In addition, we can consider simultaneously conjugating the connection matrices adjacent to $b$ or $w$ as 
\begin{equation*}
    \widehat{\phi_{b w}} = g^{-1} \phi_{b,w}
    \quad\quad\text{or}\quad\quad
    \widehat{\phi_{b w}} \to \phi_{b,w} g^{-1}.
\end{equation*}
These new connection matrices in these bases give a new web trace $\widehat{\tr}(\mathrm{web},\widehat{\Phi})$. The two web traces in the different bases are related as 
\begin{equation*}
    \widehat{\tr}(\mathrm{web},\widehat{\Phi}) = \det(g) \tr(\mathrm{web},\Phi).
\end{equation*}
As such, the web trace is not gauge-invariant under $\GL_N$. However, it is gauge-invariant under $\SL_N$ transformations. As such, restricting to $\SL_N$ connections and gauge transformations, $\phi_{bw} \in \SL_N$, will make the trace gauge-invariant. Although, many formulas and considerations in this paper don't rely on any notions of gauge-invariance and can be expressed for arbitrary $\phi_{bw} \in \GL_N$. 

Now, given a multiweb that lives on a bipartite graph in the surface with multiplicities $\{m_e\}$ on edges $\{e\}$, there is an associated web that comes from splitting edge $e=(b,w)$ into $m_e$ parallel edges, all with the matrix $\phi_{bw}$ on each edge. We will define the trace of the associated multiweb as:
\begin{equation}
    \tr(\mathrm{multiweb},\Phi) = \frac{\tr(\mathrm{web},\Phi)}{\prod_e m_e!}.
\end{equation}

A key intermediate result of~\cite{dksWebs2022} allows us to express traces of multi-webs in terms of dimer covers on the so-called `blowup' of the graph. The blowup graph is defined as the graph whose vertices and edges are given by
\begin{equation*}
\begin{split}
    \text{vertices} &= \{b^{(i)}, w^{(i)} | b,w \, \text{black, white vertices}, i \in \{1,\cdots,N\}\} \\
    \text{edges} &= \{(b^{(i)}, w^{(j)}) | (b,w) = \text{edge in graph}, i,j \in \{1,\cdots,N\}\}.
\end{split}
\end{equation*}
See Fig.~\ref{fig:blowup_graph_and_match} for a depiction.

\begin{figure}[h!]
  \centering
  \includegraphics[width=\linewidth]{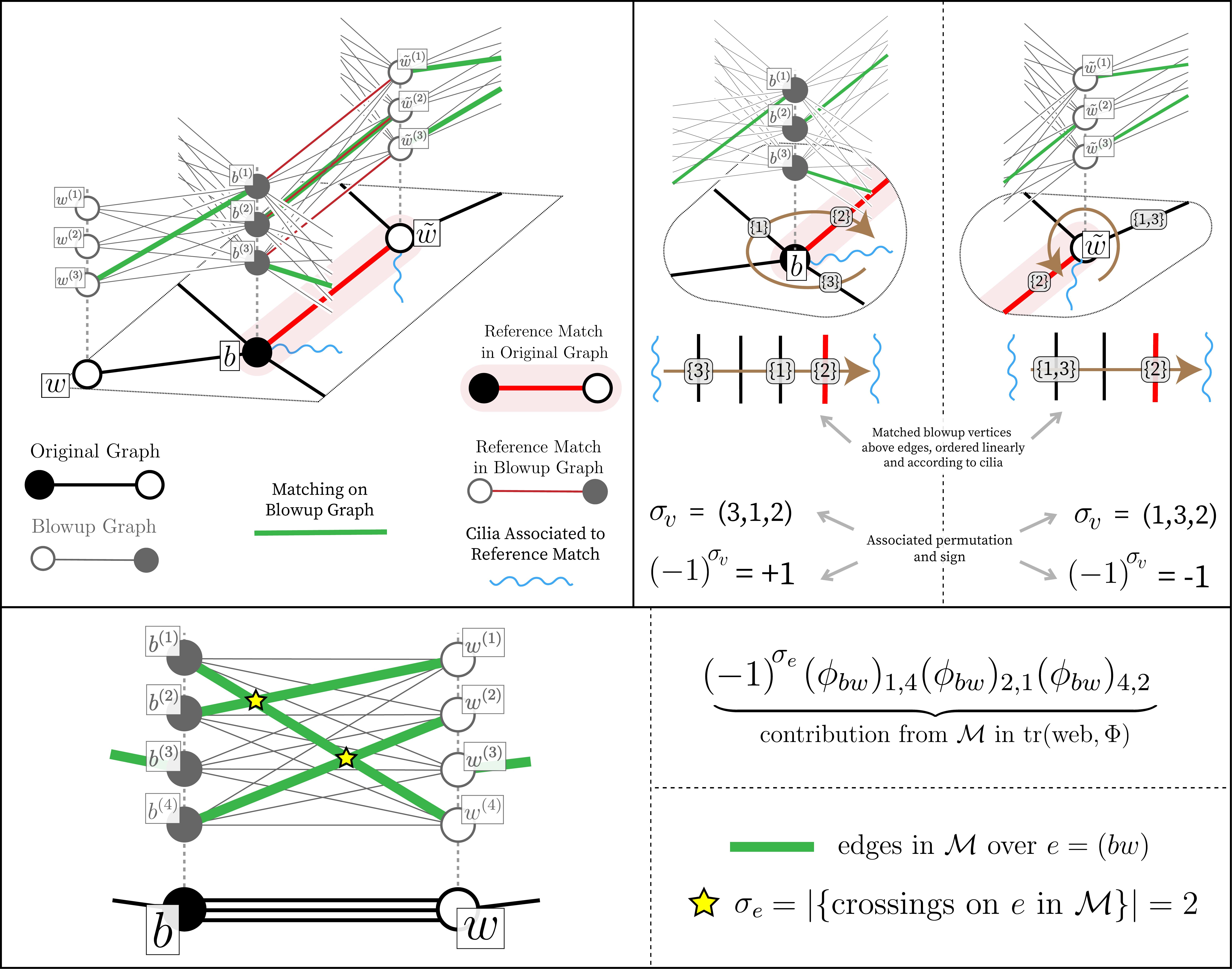}
  \caption{
    (Top-Left) Depiction of blowup graph for $N=3$, as well as the lift of a reference matching on the original graph and the associated ciliation. 
    (Top-Right) Depiction of the vertex permutations $\{\sigma_v\}$ and their associated signs contributed to the web trace. 
    (Right) Contribution to web trace from a matching $\mathcal{M}$ of the blowup graph, where the $N=4$ and the edge $e=(b,w)$ is has multiplicity 3 in the web.
  }
  \label{fig:blowup_graph_and_match}
\end{figure}

The important formula we will be comparing to in this paper is, for a multiweb with edge multiplicities $\{m_e\}$,
\begin{equation} \label{eq:multiweb_trace}
\begin{split}
    &\tr(\mathrm{multiweb},\Phi) 
    = 
    \sum_{
        \substack{
            \text{matchings } \mathcal{M} \\
            \text{ of blow-up graph with} \\
            \text{$m_e$ edges above $e$ in graph}
        }
    }
    (-1)^r
    \prod_{ (b^{(i)}, w^{(j)} ) \in \mathcal{M}}
    (\phi_{b, w})_{i,j},
    \\
    &\quad \text{ where} \quad
    (-1)^r
    = 
    \prod_{\substack{\text{edges } e = (b,w) \\ \text{in graph} }} (-1)^{\sigma_e}
    \prod_{\substack{\text{vertices } v\\ \text{in graph} }} (-1)^{\sigma_v}
\end{split}
\end{equation}
and the permutations $\{\sigma_e\}$ and $\{\sigma_v\}$ are defined as follows. For an edge $e=(b,w)$ with multiplicity $m_e$ in the multiweb, a perfect matching of the blowup graph with $m_e$ edges lying above $e$ will connect edges $b^{(i_1)} \leftrightarrow w^{(j_1)}, \cdots, b^{(i_{m_e})} \leftrightarrow w^{(j_{m_e})}$. This defines a natural permutation on $\{1,\cdots,m_e\}$ by considering the linear orders of the $\{i_1,\cdots,i_{m_e}\}$, $\{j_1,\cdots,j_{m_e}\}$. Equivalently, (overloading $\sigma_e$ to refer to both the permutation and its inversion number), we'd have:
\begin{equation}
    \sigma_e = \#\Big\{i_{a} < i_{a'} | b^{(i_a)} \leftrightarrow w^{(j_a)}, b^{(i_{a'})} \leftrightarrow w^{(j_{a'})} \in \mathcal{M}  \,\,\text{ and }\,\,  j_{a} > j_{a'}\Big\}.
\end{equation}
For any vertex $v$ of degree $d(v)$ associated edges $e_1(v),\cdots,e_{d(v)}(v)$ in the order determined by the ciliation, say that the linear order of the $m_{e_k(v)}$ vertices blowup vertices $\{v^{(1)},\cdots,v^{(N)}\}$ that are involved in matchings above $e_k(v)$ are $v^{(i_{k,1})}, \cdots, v^{(i_{k,m_{e_k(v)}})}$ with $i_{k,1} < \cdots < i_{k,m_{e_k(v)}}$. Then $\sigma_v$ is the permutation:
\begin{equation}
    \sigma_v = 
    (1, \cdots, N)
    \mapsto 
    \left(
        i_{1,1}, \cdots, i_{1,m_{e_1(v)}}, 
        \cdots, 
        i_{k,1}, \cdots, i_{1,m_{e_k(v)}}, 
        \cdots,
        i_{d(v),1}, \cdots, i_{d(v),m_{e_{d(v)}(v)}}
    \right).
\end{equation}
See again Fig.~\ref{fig:blowup_graph_and_match} for depictions of $\sigma_e,\sigma_v$ and both the sign and connection contributions that come from a matching on the blowup.

We note that in~\cite{dksWebs2022}, the analogous formula for the sum is written in terms of so-called `half-edge colorings'. A half-edge coloring is precisely equivalent to a dimer matching on the blowup graph.

The proof of Eq.~\eqref{eq:multiweb_trace} follows almost immediately from the definition and requisite bookkeeping. For the case of all $m_e = 1$ (i.e. the multiweb is just a web with no multiple edges), terms in the tensor contraction are exactly the same as terms in the expansion over blowup matchings. In the case of edge multiplicities, there are $\prod_{e} m_e!$ equal terms in the tensor contraction that correspond to a single blowup matching, which explains the need to normalize by $\frac{1}{\prod_{e} m_e!}$.

To conclude this section, we will explain how to think of the web traces in terms of the more familiar objects of traces of monodromies of loops. First, consider a multiweb whose vertices consist of a $\mathrm{loop} = b_1 \to w_1 \to b_2 \to w_2 \to \cdots \to b_n \to w_n \to b_1$ and for which each $b_i \to w_i$ has multiplicity $k$ and each $w_i \to b_{i+1}$ has multiplicity $N-k$. This loop has an associated monodromy 
\begin{equation}
    \mathrm{mon(loop)} = \phi_{b_1,w_1} \phi_{w_1,b_2} \cdots \phi_{b_n, w_n} \phi_{w_n,b_1}.
\end{equation}
Call this multiweb $\mathrm{loop}_k$. We'll have (up to an overall $\pm$ sign depending on the ciliation)
\begin{equation}
    \tr(\mathrm{loop}_k, \Phi) = \pm \tr\left(\Lambda^k \mathrm{mon(loop)}\right)
\end{equation}
where $\Lambda^k$ is the $k$th exterior power. 
See Fig.~\ref{fig:tr_loop_kth_ext} for depictions.

\begin{figure}[h!]
  \centering
  \includegraphics[width=0.7\linewidth]{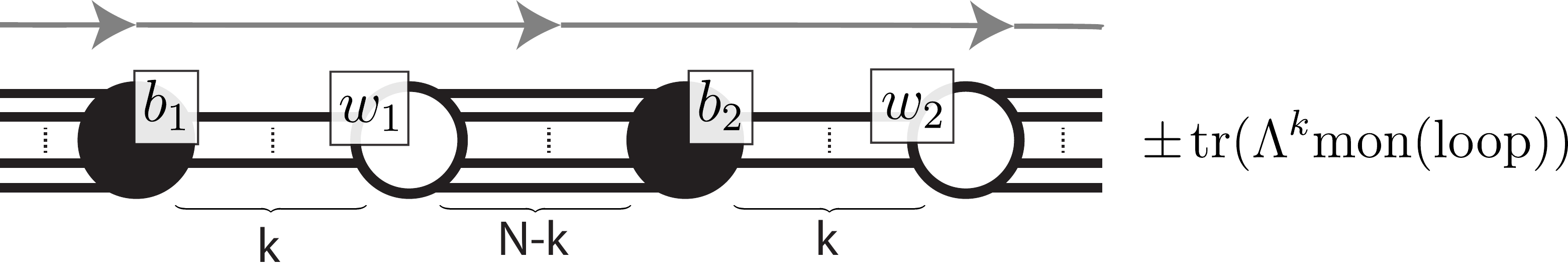}
  \caption{
    A multiweb consisting of edge multiplicites of $k$, $N-k$ in a $\mathrm{loop} = b_1 \to w_1 \to b_2 \to w_2 \to \cdots$ evaluates to $\pm \tr(\Lambda^k \mathrm{mon(loop)})$.
  }
  \label{fig:tr_loop_kth_ext}
\end{figure}

In general, webs on a punctured surface with an $\SL_N$ span `regular functions on the $\SL_N$ character variety of the punctured surface'. In other words, webs are `algebraic' functions of the moduli space of $\SL_N$ connections on the punctured surface. It is known that another generating set of such functions are traces of monodromies of loops~\cite{PROCESI1976306}. Indeed, one can use skein relations (see e.g.~\cite{sikoraWebs01}) to express any web trace as sums over products of traces of monodromies of loops. In the bipartite graph / multiweb setting, we think of the punctured surface as having one puncture per face, and a general multiweb trace can be written as a linear combination of products of traces of monodromies of loops in the surface.

We now note that for the rest of this article, we will exclusively be dealing with multiwebs. As such, we will often refer to multiwebs as webs.

\subsection{Geometric Setup on Graphs} \label{sec:geometric_setup}
Now, we review how a choice of Kasteleyn signs plus a reference dimer matching corresponds to a spin structure. We will use them to get higher genus versions of Kasteleyn's formula and web trace formulas with a spin structure are obtained for free with some geometric setup and background.

Consider a bipartite graph embedded in a surface $\Sigma$ so that every face is contractible. If we endow the graph with a reference matching $\mathcal{R}$, one can construct a nonvanishing vector field along the graph as in Fig.~\ref{fig:vectorFieldKasteleyn}. This construction of vector field will have an even-index singularity on each face of degree $2 \,  (\mathrm{mod} \, 4)$ and an odd index singularity for degree $0 \, (\mathrm{mod} \, 4)$.

\begin{figure}[h!]
  \centering
  \includegraphics[width=0.95\linewidth]{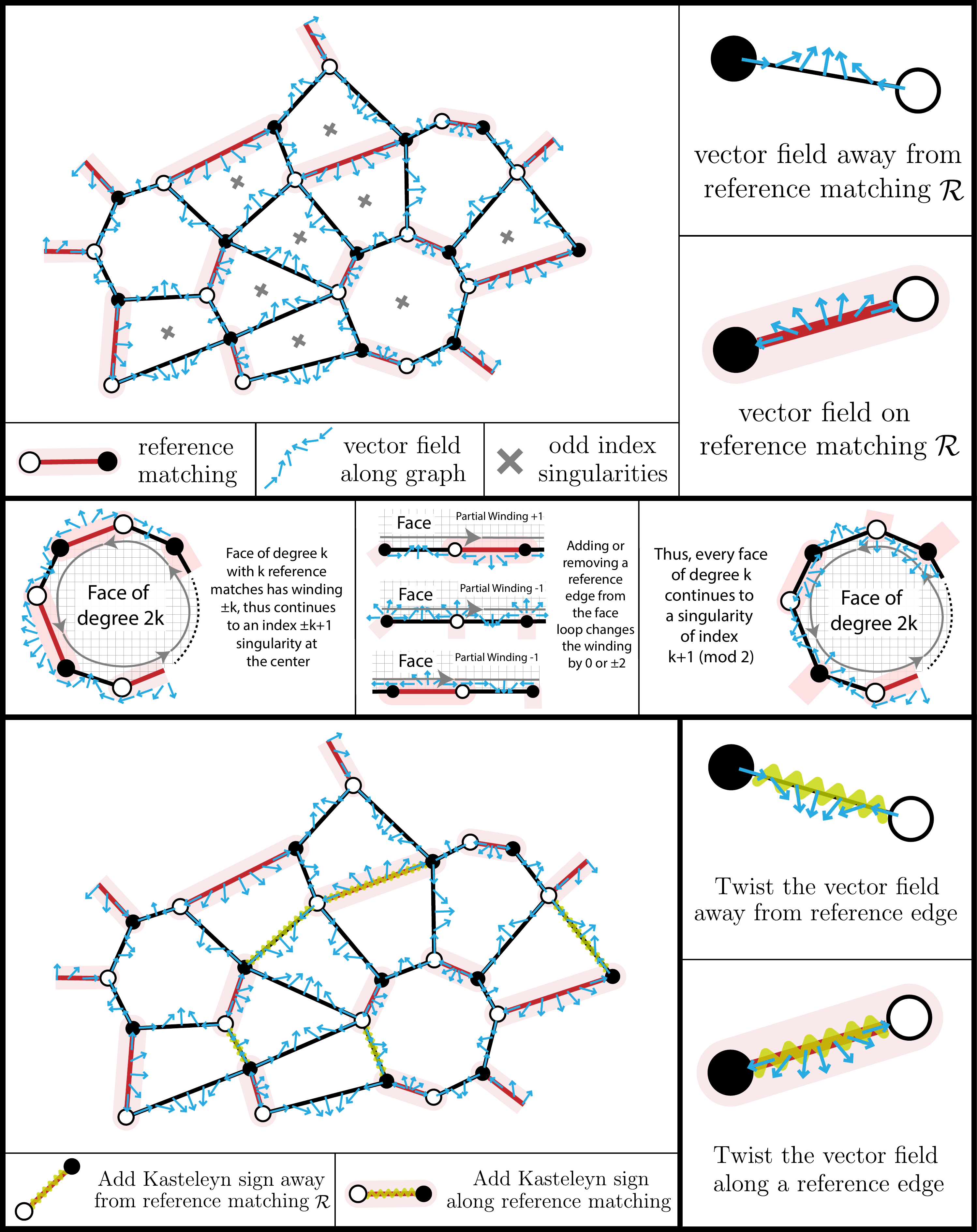}
  \caption{
    (Top) Given a reference dimer matching, one can construct a continuous vector field in a neighborhood of the surface graph. In general this vector field can be extended to the faces of the graph with an even index singularity on a face with $2 \,  (\mathrm{mod} \, 4)$ sides and an odd index singularity on a face with $0 \, (\mathrm{mod} \, 4)$ sides.
    (Middle) Explanation of why for every face of degree $2k$, the vector field to a singularity of index $k+1 \, (\mathrm{mod} \, 2)$ on the interior.
    (Bottom) In general, one can choose a collection of graph edges between pairs of odd-index singularities to twist the vector field, returning a new vector field with even-index singularities everywhere. These twists correspond to Kasteleyn signs.
  }
  \label{fig:vectorFieldKasteleyn}
\end{figure}

Also as in Fig.~\ref{fig:vectorFieldKasteleyn}, one can \textit{twist} the vector field along a collection of edges so that the twisted vector field extends to even-index singularities on each face. Generally, this will always be possible on a surface graph since there are always an even number of odd-index singularities on an oriented surface. For an edge
${\bf e} = (b w)$ define the \textit{Kasteleyn sign}
\begin{equation}
    \epsilon_{{\bf e}}
    =
    \epsilon_{b w}
    =
    \begin{cases}
        -1, &\quad\text{ if vector field twisted along edge}
        \\
        +1, &\quad\text{ otherwise} 
    \end{cases}.
\end{equation}
We want to choose the Kasteleyn signs so that the twisted vector field extends to even-index singularities going around each face.
As reviewed in Appendix~\ref{app:spin_structures}, this condition will mean that the twisted vector field corresponds to a \textit{spin-structure}. Now, consider a face in the graph. Note that every edge $(bw)$ in the face with $\epsilon_{b,w} = -1$ changes the index of the singularity on the face by $\pm 1$. This can be seen by computing the winding number of the vector field with respect to the tangent of the curve going around the face. Twisting the edge by a Kasteleyn sign will change the local winding angle across an edge from $\pm \pi \to \mp \pi$, which means that the winding changes by a single unit $\pm 2\pi$. In particular, to arrange even index singularities on a face, we will need an even number of $\epsilon_{b,w} = -1$ on any face with degree $2 \,  (\mathrm{mod} \, 4)$ and an odd number on faces with degree $0 \,  (\mathrm{mod} \, 4)$. 
This condition corresponds exactly to the Kasteleyn sign condition
\begin{equation}
    \prod_{(bw) \in \mathrm{face}} \epsilon_{b,w}
    =
    \begin{cases}
        +1 \quad \text{ if face has degree } \, 2 \,  (\mathrm{mod} \, 4) \\
        -1 \quad \text{ if face has degree } \, 0 \,  (\mathrm{mod} \, 4)
    \end{cases}.
\end{equation}

In conclusion, a choice of reference dimer matching together with a choice of Kasteleyn signs $\epsilon_{b,w}$ give a construction of a spin structure on the surface.

\section{Rank-1 Dimers} \label{sec:rank_1_dimer}

Here, we review the results of~\cite{cimasoniReshetikhin2007spin} which proves Kasteleyn's theorem on Riemann surfaces. We will assume a bipartite graph embedded in a surface together with a reference matching. The point is that a Kasteleyn determinant for a single spin structure is a \textit{signed} sum over dimer covers where the sign is related to the topological class of the dimer cover. We will explain the specifics of this statement here, then in the Subsections~\ref{sec:rank_1_dimer_planar},\ref{sec:rank_1_dimer_surfaces} explain applications to getting unsigned sums over dimer covers in the planar and surface case respectively.

The statement we will prove is as follows.
\begin{thm}
    Let the reference matching $\mathcal{R}$ and Kasteleyn signs $\epsilon$ correspond to a spin structure $\eta$. Denote the corresponding Kasteleyn matrix $K_\eta$ (which is a slight abuse of notation since there can be many choices of $\epsilon$ corresponding to isomorphic spin structures). Then
    \begin{equation}
        \det(K_\eta) 
        = \pm
        \sum_{\substack{\text{dimer covers} \\ \mathcal{M}} }
        (-1)^{q_\eta(\mathcal{M} \sqcup \mathcal{R})}
    \end{equation}
    In the above, we note that the reference match $\mathcal{R}$ together with another matching $\mathcal{M}$ can be be combined to form a collection of loops $\mathcal{M} \sqcup \mathcal{R}$, which can be mapped to a homology class in $H_1(\Sigma,\Z_2)$. The sign $(-1)^{q_\eta(\cdots)}$ is defined in Sec.~\ref{app:quad_form_arf} of the Appendix.
\end{thm}

\begin{proof}
We will use the language of Grassmann variables to show the statement.
Now, we will take the matrix in question to be the Kasteleyn matrix $(K_\eta)_{b,w} = \epsilon_{b,w}$, where the $\epsilon_{b,w}$ are the Kasteleyn signs.

Let's expand the determinant.
\begin{equation}
    \det(K_\eta) 
    = 
    \sum_{\sigma \in \mathrm{Sym}(n)} (-1)^{\sigma} (K_\eta)_{b_1 w_{\sigma(1)}} \cdots (K_\eta)_{b_n w_{\sigma(n)}}
    =
    \sum_{\substack{\sigma \in \mathrm{Sym}(n) \text{ s.t.} \\ (b_1,w_{\sigma(1)}) \cdots (b_n,w_{\sigma(n)}) \\ \text{is a dimer cover} \, \mathcal{M}}}
    (-1)^{\sigma} \epsilon_{b_1, w_{\sigma(1)}} \cdots \epsilon_{b_n, w_{\sigma(n)}}.
\end{equation}
We will argue that each of the 
\begin{equation*}
    (-1)^{\sigma} \epsilon_{b_1, w_{\sigma(1)}} \cdots \epsilon_{b_n, w_{\sigma(n)}}
    =
    \pm (-1)^{q_\eta(\mathcal{M} \sqcup \mathcal{R})}
\end{equation*} 
for a uniform $\pm$ sign. Consider the special case $\mathcal{M} = \mathcal{R}$ (i.e. if $\mathcal{M}$ is the reference matching). Then, $\mathcal{M} \sqcup \mathcal{R}$ will be a homologically trivial collection of doubled edges so that $(-1)^{q_\eta(\mathcal{M} \sqcup \mathcal{R})} = +1$. 

In the language of Grassmann variables, given a set of variables $\psi_w, \overline{\psi}_b$ indexed by white, black vertices respectively, this condition can be rephrased as follows. Say that the permutation $\tau$ corresponds to $\mathcal{R}$ and $\sigma$ corresponds to $\mathcal{M}$. Then, this theorem is equivalent to showing
\begin{equation} \label{eq:kast_grassmann_sign_uniform}
\begin{split}
    &=(-\overline{\psi}_{b_1} \epsilon_{b_1,w_{\tau(1)}}   \psi_{w_{\tau(1)}}) 
    \cdots 
    (-\overline{\psi}_{b_n} \epsilon_{b_n,w_{\tau(n)}}   \psi_{w_{\tau(n)}}) 
    =
    (-1)^{\tau} \epsilon_{b_1, w_{\tau(1)}} \cdots \epsilon_{b_n, w_{\tau(n)}}
    \psi_{w_1} \overline{\psi}_{b_1} \cdots \psi_{w_n} \overline{\psi}_{b_n}
    \\
    &=\pm \psi_{w_1} \overline{\psi}_{b_1} \cdots \psi_{w_n} \overline{\psi}_{b_n}
    \\
    &\quad\quad\quad\quad\quad\quad\quad\quad
     \quad\quad\quad\quad\quad\quad\quad\quad \text{and}
    \\
    &(-\overline{\psi}_{b_1} \epsilon_{b_1,w_{\sigma(1)}} \psi_{w_{\sigma(1)}}) 
    \cdots 
    (-\overline{\psi}_{b_n} \epsilon_{b_n,w_{\sigma(n)}} \psi_{w_{\sigma(n)}}) 
    = 
    (-1)^{\sigma} \epsilon_{b_1, w_{\sigma(1)}} \cdots \epsilon_{b_n, w_{\sigma(n)}}
    \psi_{w_1} \overline{\psi}_{b_1} \cdots \psi_{w_n} \overline{\psi}_{b_n}
    \\
    &= \pm (-1)^{q_\eta(\mathcal{M} \sqcup \mathcal{R})} \psi_{w_1} \overline{\psi}_{b_1} \cdots \psi_{w_n} \overline{\psi}_{b_n}
\end{split}
\end{equation}
for a uniform $\pm$ sign. Another equivalent statement can be found by considering another set of Grassmann variables $\overline{\psi}_{w}, \psi_b$ indexed by white, black vertices. The equivalent condition is that
\begin{equation} \label{eq:kast_grassmann_sign_equiv}
\begin{split}
    &(-\overline{\psi}_{b_1} \epsilon_{b_1,w_{\sigma(1)}} \psi_{w_{\sigma(1)}}) 
    \cdots 
    (-\overline{\psi}_{b_n} \epsilon_{b_n,w_{\sigma(n)}} \psi_{w_{\sigma(n)}})
    (-\psi_{b_1} \epsilon_{b_1,w_{\tau(1)}} \overline{\psi}_{w_{\tau(1)}}) 
    \cdots 
    (-\psi_{b_n} \epsilon_{b_n,w_{\tau(n)}} \overline{\psi}_{w_{\tau(n)}})
    \\
    &=(-\overline{\psi}_{b_1} \epsilon_{b_1,w_{\sigma(1)}} \psi_{w_{\sigma(1)}}) 
    \cdots 
    (-\overline{\psi}_{b_n} \epsilon_{b_n,w_{\sigma(n)}} \psi_{w_{\sigma(n)}})
    (+\overline{\psi}_{w_{\tau(1)}} \epsilon_{b_1,w_{\tau(1)}} \psi_{b_1}) 
    \cdots 
    (+\overline{\psi}_{w_{\tau(n)}} \epsilon_{b_n,w_{\tau(n)}} \psi_{b_n} )
    \\
    &=
    + (-1)^{q_\eta(\mathcal{M} \sqcup \mathcal{R})}
    \psi_{w_1} \overline{\psi}_{b_1} \cdots \psi_{w_n} \overline{\psi}_{b_n}
    \overline{\psi}_{w_1} \psi_{b_1} \cdots \overline{\psi}_{w_n} \psi_{b_n}
    = (-1)^{q_\eta(\mathcal{M} \sqcup \mathcal{R})}    \psi_{b_1} \overline{\psi}_{b_1} \psi_{w_1} \overline{\psi}_{w_1} 
    \cdots
    \psi_{b_n} \overline{\psi}_{b_n} \psi_{w_n} \overline{\psi}_{w_n}.
\end{split}
\end{equation}
The second equality follows because the uniform $\pm$ sign for $\sigma$ and $\tau$ in Eq.~\eqref{eq:kast_grassmann_sign_uniform} are the same for both sets of variables. The third equality follows since the two expressions differ by an even number of transpositions of Grassmann variables.

The following Lemma will be instrumental in showing Eq.~\eqref{eq:kast_grassmann_sign_equiv} and thus the theorem.
\begin{lem} \label{lem:grassmann_equals_winding}
    Let ${\bf e}_1 \to {\bf e}_2 \to \cdots \to {\bf e}_{2m-1} \to {\bf e}_{2m} \to {\bf e}_1$ be a loop of edges as follows.
    First, any vertex appears at most once. 
    Next, the odd edges ${\bf e}_1,{\bf e}_3,\cdots$ should be taken to be in the reference matching. 
    And, the even edges ${\bf e}_2,{\bf e}_4,\cdots$ should be \textit{away from} the reference matching. 
    Denote ${\bf e}_{2k-1} = (w_{\ell_k},b_{\ell_k})$, ${\bf e}_{2k} = (b_{\ell_k},w_{\ell_{k+1}})$ as the white,black vertices in each edge, identifying $w_{\ell_{m+1}}=w_{\ell_1}$, $b_{\ell_{m+1}}=b_{\ell_1}$. Then,
    \begin{equation}
        \left(
        \prod_{k=1}^{m}
            \left(
                -\epsilon_{{\bf e}_{2k}} 
                \overline{\psi}_{b_k} \psi_{w_{k+1}} 
            \right)
            \left( 
                \epsilon_{{\bf e}_{2k-1}} 
                \overline{\psi}_{w_k} \psi_{b_k}     \right)
            \right)
        =
        (-1)^{1+\mathrm{wind(loop)}}
        \psi_{b_{\ell_1}} \overline{\psi}_{b_{\ell_1}}
        \psi_{w_{\ell_1}} \overline{\psi}_{w_{\ell_1}}
        \cdots
        \psi_{b_{\ell_1}} \overline{\psi}_{b_{\ell_1}}
        \psi_{w_{\ell_1}} \overline{\psi}_{w_{\ell_1}}
    \end{equation}
    where $\mathrm{wind(loop)}$ is the winding of the tangent with respect to the vector field constructed in Sec.~\ref{sec:geometric_setup}.
\end{lem}
Before showing the Lemma, we explain why it implies the theorem. The point is that we can rearrange (at no sign cost) the product in the beginning of Eq.~\eqref{eq:kast_grassmann_sign_equiv} into a product over the loops and doubled edges between $\mathcal{M}$ and $\mathcal{R}$. As such, each loop will contribute a factor of $(-1)^{1+\mathrm{wind(loop)}}$ after rearranging the loops. Then, the product of the $\psi_{b} \overline{\psi}_b$ and $\psi_{w} \overline{\psi}_w$ factors can again be rearranged at no sign cost to give $\psi_{b_1} \overline{\psi}_{b_1} \psi_{w_1} \overline{\psi}_{w_1} \cdots$, times an overall factor of 
\begin{equation*}
    \prod_{\mathrm{loops} \in \mathcal{M} \sqcup \mathcal{R}} (-1)^{1+\mathrm{wind(loop)}} = (-1)^{q_\eta(\mathcal{M} \sqcup \mathcal{R})}.
\end{equation*}
These loops will manifestly have no intersections when all embedded in the plane. And above, we note that the product of the $(-1)^{1+\mathrm{wind(loop)}}$ for a collection of loops with no crossings is by definition $ (-1)^{q_\eta(\mathcal{M} \sqcup \mathcal{R})}$, as explained in Sec.~\ref{app:quad_form_arf}. 

We now proceed to prove the Lemma.
\begin{proof}[Proof of Lemma~\ref{lem:grassmann_equals_winding}]
    We briefly handle the case of a doubled edge, i.e. $m=1$ separately. This degenerate case can be thought of as a small non-self-intersecting loop around the edge, so that $(-1)^{1+\mathrm{wind(loop)}}=1$. Later in Sec.~\ref{sec:rank_N_dimers} while discussing higher-rank dimers, we will see a more systematic way to think about embedding the loop in this way. For $m=1$ and loop involving $b,w$ we have
    \begin{equation}
        (-\epsilon_{b,w} \overline{\psi}_b \psi_w) (+\epsilon_{b,w} \overline{\psi}_w \psi_b) 
        = 
        \psi_b \overline{\psi}_b \psi_w \overline{\psi}_w
        = 
        (-1)^{1+\mathrm{wind(loop)}}
        \psi_b \overline{\psi}_b \psi_w \overline{\psi}_w
    \end{equation}
    
    Now, we consider $m>1$. Note that we only need to prove this statement in the case that all $\epsilon_{\cdots} = +1$. In particular, changing the sign of any $\epsilon_{w b}$ will change the left-side by a sign and also change the winding of the loop by $\pm 1$ giving an extra minus sign on the right-side. We proceed to verify the case of untwisted windings, with $\epsilon_{\cdots}=1$.
    
    In this case, one can check that $\mathrm{wind(loop)} = m$. On the other hand, we'd have 
    \begin{equation}
    \begin{split}
        \left(\overline{\psi}_{w_{\ell_1}} \psi_{b_{\ell_1}}\right)
        \left(-\overline{\psi}_{b_{\ell_1}} \psi_{w_{\ell_2}} \right)
        \cdots
        \left(\overline{\psi}_{w_{\ell_m}} \psi_{b_{\ell_m}}\right)
        \left(-\overline{\psi}_{b_{\ell_m}} \psi_{w_1} \right)
        = (-1)^{m+1}
        \left(
            \psi_{b_1} \overline{\psi}_{b_1} 
            \psi_{w_1} \overline{\psi}_{w_1}
            \cdots
            \psi_{b_m} \overline{\psi}_{b_m}
            \psi_{w_m} \overline{\psi}_{w_m} 
        \right)
    \end{split}
    \end{equation}
    These signs indeed match. See Fig.~\ref{fig:grassmann_equals_wind_ex} for an example. 
\end{proof}
\begin{figure}[h!]
  \centering
  \includegraphics[width=\linewidth]{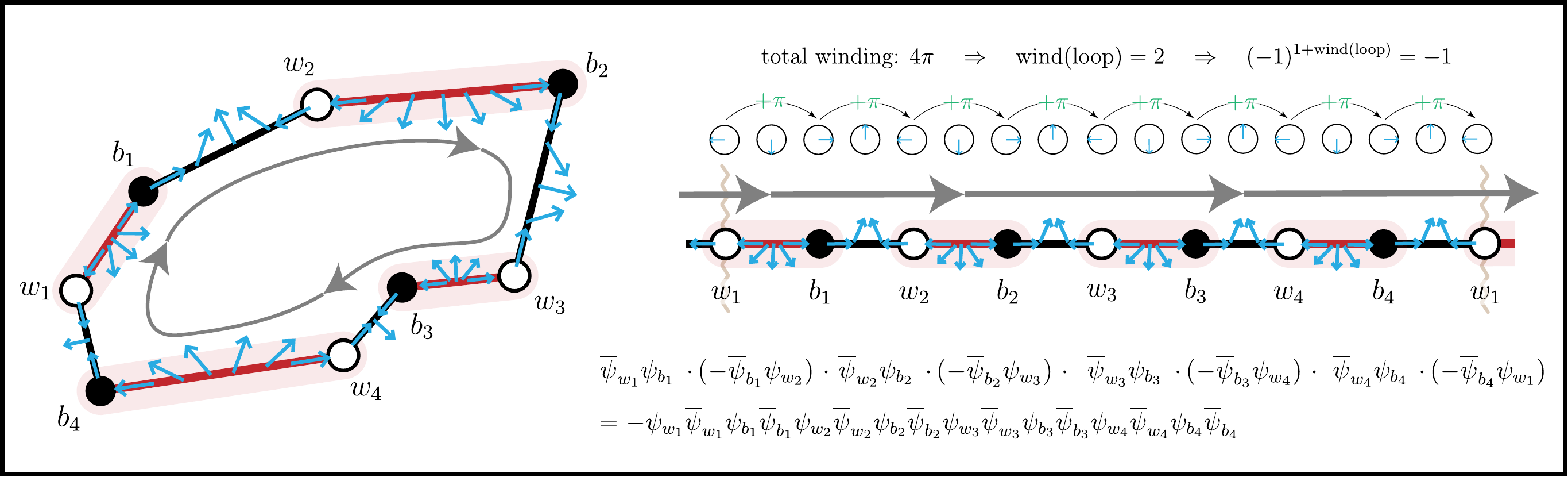}
  \caption{
    Example $m=4$ base case for showing that the Grassmann integral is related to the winding of the tangent vector along the curve as in Lemma~\ref{lem:grassmann_equals_winding}.
  }
  \label{fig:grassmann_equals_wind_ex}
\end{figure}
This completes the proof.
\end{proof}

\subsection{Planar Case} \label{sec:rank_1_dimer_planar}
Note that since the plane is topologically trivial, there is only one spin structure and $(-1)^{q_\eta(\mathcal{M} \sqcup \mathcal{R})} = +1$ for every matching $\mathcal{M}$. This means that 
\begin{equation}
    \det(K_\eta) 
    = \pm
    \sum_{\substack{\text{dimer covers} \\ \mathcal{M}} }
    1
    = \pm \# \{\text{dimer covers}\}
\end{equation}
which is exactly the original Kasteleyn theorem.

\subsection{On Surfaces} \label{sec:rank_1_dimer_surfaces}
On a Riemann surface $\Sigma$ of genus $g$, there are $2^{2g}$ spin structures. We would like to use the information of all of these spin structures to compute the number of dimer covers. Henceforth, we'll assume that all spin structures $\eta$ are constructed using the same reference matching and only differ in their choice of Kasteleyn signs. Moreover, after appropriate gauge transformations, we can always choose the Kasteleyn signs to have the same values on reference edges over all spin structures $\eta$. This will make the $\pm$ sign in front of 
\begin{equation*}
    \det(K_\eta) 
    = 
    \pm
    \sum_{\substack{\text{dimer covers} \\ \mathcal{M}} }
    (-1)^{q_\eta(\mathcal{M} \sqcup \mathcal{R})}
\end{equation*}
be the same for all $\eta$. This is because this sign equals $(-1)^{\tau} \epsilon_{b_1, w_{\tau(1)}} \cdots \epsilon_{b_n, w_{\tau(n)}}$ for $\tau$ the permutation corresponding to the reference edge. 

We will use the fact Lemma~\ref{lem:inv_quad_form} that the quadratic form matrix $A_{\eta,\xi} = (-1)^{q_\eta(\xi)}$ is invertible with inverse $(A^{-1})_{\xi,\eta} = \frac{1}{2^{2g}} (-1)^{q_\eta(\xi)}$. Let us decompose the sum over homology classes. For $\xi \in H_1(\Sigma,\Z_2)$, define
\begin{equation}
    Z_\xi
    :=
    \sum_{\substack{\text{dimer covers} \\ \mathcal{M} \\ \text{s.t. } [\mathcal{M} \sqcup \mathcal{R}] = [\xi] \in H_1(\Sigma,\Z_2) } }
    1
\end{equation}
so that $Z := \sum_{\xi \in H_1(\Sigma,\Z_2)} Z_\xi = \#\text{\{dimer covers\}}$.

We have
\begin{equation}
    \det(K_\eta) 
    = 
    \pm
    \sum_{\xi \in H_1(\Sigma,\Z_2)}
    (-1)^{q_\eta(\xi)}
    \sum_{\substack{\text{dimer covers} \\ \mathcal{M} \\ \text{s.t. } [\mathcal{M} \sqcup \mathcal{R}] = [\xi] \in H_1(\Sigma,\Z_2) } }
    1
    =
    \pm
    \sum_{\xi \in H_1(\Sigma,\Z_2)}
    A_{\eta,\xi} Z_\xi
\end{equation}
This means that we can write
\begin{equation}
    Z_\xi
    = 
    \pm
    \sum_{\substack{\eta \, \text{spin} \\ \text{structure}}}
    (A^{-1})_{\xi,\eta} \det(K_\eta)
    = 
    \pm \frac{1}{2^{2g}}
    \sum_{\substack{\eta \, \text{spin} \\ \text{structure}}}
    (-1)^{q_\eta(\xi)} \det(K_\eta).
\end{equation}
Finally, we get
\begin{equation}
    Z 
    = 
    \sum_{\xi \in H_1(\Sigma,\Z_2)} Z_\xi
    =
    \pm
    \frac{1}{2^{2g}}
    \sum_{\xi \in H_1(\Sigma,\Z_2)} 
    \sum_{\substack{\eta \, \text{spin} \\ \text{structure}}}
    (-1)^{q_\eta(\xi)} \det(K_\eta).
\end{equation}
Now, we can use the definition of the Arf invariant (see Sec.~\ref{app:quad_form_arf}) to write:
\begin{equation}
    Z 
    =
    \pm
    \frac{1}{2^g} 
    \sum_{\substack{\eta \, \text{spin} \\ \text{structure}}}
    \mathrm{Arf}(\eta) \det(K_\eta).
\end{equation}

\section{Rank-$N$ Dimers} \label{sec:rank_N_dimers}
The goal of this section will be to prove the following deformation of the main result of~\cite{dksWebs2022}. In the rest of this section, $\{\epsilon_{bw}\}$ be a choice of Kasteleyn signs corresponding to spin structure $\eta$, $\mathcal{R}$ be a reference matching in the graph and $\tilde{\mathcal{R}}$ be the lift of $\mathcal{R}$ in the blowup graph embedded in surface $\Sigma$.
\begin{thm} \label{thm:dks_on_surfaces}
\begin{equation}
    \det(K_\eta(\Phi)) = \pm \sum_{\substack{\text{web} \\ \text{configs}}} \, \prod_{\text{webs}} \tr_\eta(\text{web},\Phi),
\end{equation}
where we define each $\tr_\eta(\text{web},\Phi)$ in terms of matchings on the blow-up graph.
Moreover, we define 
\begin{equation}
\begin{split}
    &\tr_{\eta}(\mathrm{web},\Phi) 
    = 
    \sum_{
        \substack{
            \text{matchings } \mathcal{M} \\
            \text{of blow-up graph with} \\
            \text{compatible edge multiplicities}
        }
    }
    (-1)^s
    \prod_{ (b^{(i)}, w^{(j)} ) \in \mathcal{M}}
    (\phi_{b, w})_{i,j},
    \\
    &\quad \text{ where} \quad
    (-1)^s 
    = 
    (-1)^{q_\eta(\mathcal{M} \sqcup \tilde{\mathcal{R}})}
    \prod_{\substack{\text{edges } e = (b,w) \\ \text{in graph} }} (-1)^{\sigma_e}
    \prod_{\substack{\text{vertices } v\\ \text{in graph} }} (-1)^{\sigma_v}
\end{split}
\end{equation}
and $\mathcal{M} \sqcup \tilde{\mathcal{R}}$ refers to the $\Z_2$ homology class consisting of the union of the matching $\mathcal{M}$ with the liftted reference matching $\tilde{R}$ in the blowup graph, projected back down to the surface.
\end{thm}
In the planar case, we have only one spin structure with $(-1)^{q_\eta(\mathcal{M} \sqcup \tilde{\mathcal{R}})} = +1$, so $\tr(\mathrm{web},\Phi) = \tr_\eta(\mathrm{web},\Phi)$, and the above theorem reduces to the original result of~\cite{dksWebs2022}.

Note that by the same arguments of Sec.~\ref{sec:rank_1_dimer_surfaces}, we can sum over spin structures weighted with certain signs to compute the sum over web traces, and the restrictions of web traces to terms lying in specific homology classes. In particular, for $\xi \in H_1(\Sigma,\Z_2)$, define
\begin{equation}
\begin{split}
    &\tr_{\xi}(\mathrm{web},\Phi) 
    := 
    \sum_{
        \substack{
            \text{matchings } \mathcal{M} \\
            \text{of blow-up graph with} \\
            \text{compatible edge multiplicities} \\
            \text{ and } [\mathcal{M} \sqcup \tilde{\mathcal{R}}] = [\xi] \in H_1(\Sigma,\Z_2)
        }
    }
    (-1)^t
    \prod_{ (b^{(i)}, w^{(j)} ) \in \mathcal{M}}
    (\phi_{b, w})_{i,j},
    \\
    &\quad \text{ where} \quad
    (-1)^t
    = 
    \prod_{\substack{\text{edges } e = (b,w) \\ \text{in graph} }} (-1)^{\sigma_e}
    \prod_{\substack{\text{vertices } v\\ \text{in graph} }} (-1)^{\sigma_v}.
\end{split}
\end{equation}
Again as in Sec.~\ref{sec:rank_1_dimer_surfaces}, we choose all Kasteleyn signs corresponding to any spin structure to have the same $\{\epsilon_{bw}\}$ on the same set of reference edges. Defining
\begin{equation}
    Z_\xi(\Phi) := \sum_{\substack{\text{web} \\ \text{configs}}} \, \prod_{\text{webs}} \tr_\xi(\text{web},\Phi)
    \quad\quad\quad\text{and}\quad\quad\quad
    Z(\Phi) := \pm \sum_{\xi} Z_\xi(\Phi) = \pm \sum_{\substack{\text{web} \\ \text{configs}}} \, \prod_{\text{webs}} \tr(\text{web},\Phi)
\end{equation}
we have by the same arguments as Sec.~\ref{sec:rank_1_dimer_surfaces} the following.
\begin{cor}
    \begin{equation}
        Z_\xi(\Phi)
        = 
        \pm \frac{1}{2^{2g}}
        \sum_{\substack{\eta \, \text{spin} \\ \text{structure}}}
        (-1)^{q_\eta(\xi)} \det(K_\eta(\Phi))
        \quad\quad\text{and}\quad\quad
        Z(\Phi) = 
        \pm \frac{1}{2^g} 
        \sum_{\substack{\eta \, \text{spin} \\ \text{structure}}}
        \mathrm{Arf}(\eta) \det(K_\eta(\Phi))
    \end{equation}
\end{cor}

To prove Theorem~\ref{thm:dks_on_surfaces}, we'll first discuss in Sec.~\ref{sec:racetrack_construction} an important `racetrack construction' which we use to draw loops of $\mathcal{M} \sqcup \tilde{\mathcal{R}}$ on the surface. Then in Sec.~\ref{sec:proof_main_thm}, we prove our main theorem above.

\subsection{Racetrack Construction} \label{sec:racetrack_construction}
A key tool we will be to draw the blowup graph directly on the surface to be able to draw loops in $\mathcal{M} \sqcup \tilde{\mathcal{R}}$ directly on the surface. We do this by introducing what will be referred to as a `racetrack construction'. In particular, given the reference dimer cover $\mathcal{R}$ and its lift in the blowup $\tilde{\mathcal{R}}$, we will find a way to embed loops in the plane so that all intersections within a matching are in general position and we can use the geometric content of Appendix~\ref{app:spin_structures} to analyze the sign factors, which are related to windings of loops in $\mathcal{M} \sqcup \tilde{\mathcal{R}}$. This projection is explained visually in Figs.~\ref{fig:racetrack},\ref{fig:racetrack_and_examples}, where Fig.~\ref{fig:racetrack} draws the racetrack graph and Fig.~\ref{fig:racetrack_and_examples} explains how to draw loops in the graph. 
We summarize the construction in words below. 

\begin{figure}[h!]
  \centering
  \includegraphics[width=\linewidth]{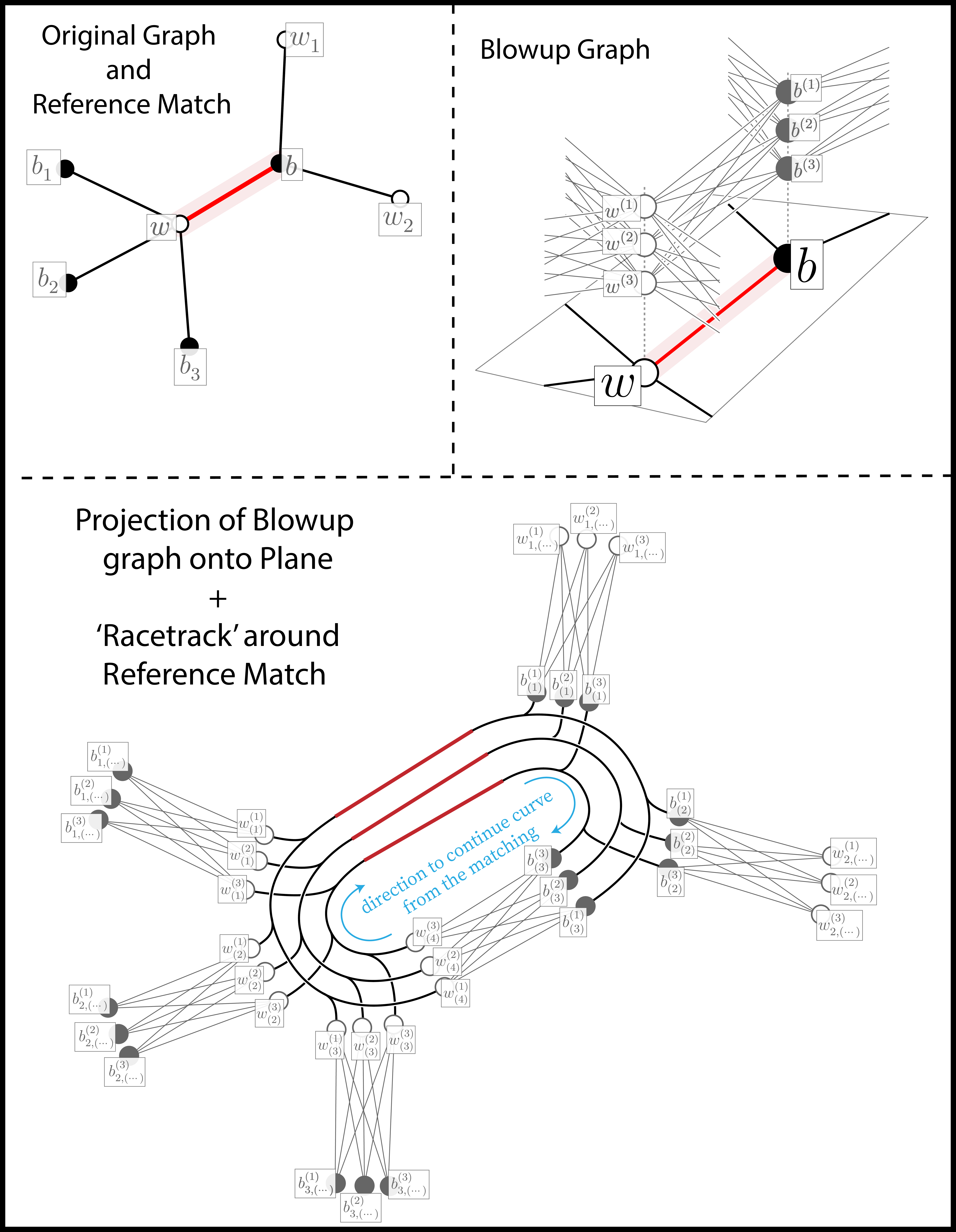}
  \caption{
    (Top-Right) Original graph (Top-Middle) Blowup graph (Bottom) Canonical projection of graph onto the plane, making a `racetrack' near each reference edge.
  }
  \label{fig:racetrack}
\end{figure}

\begin{figure}[h!]
  \centering
  \includegraphics[width=0.9\linewidth]{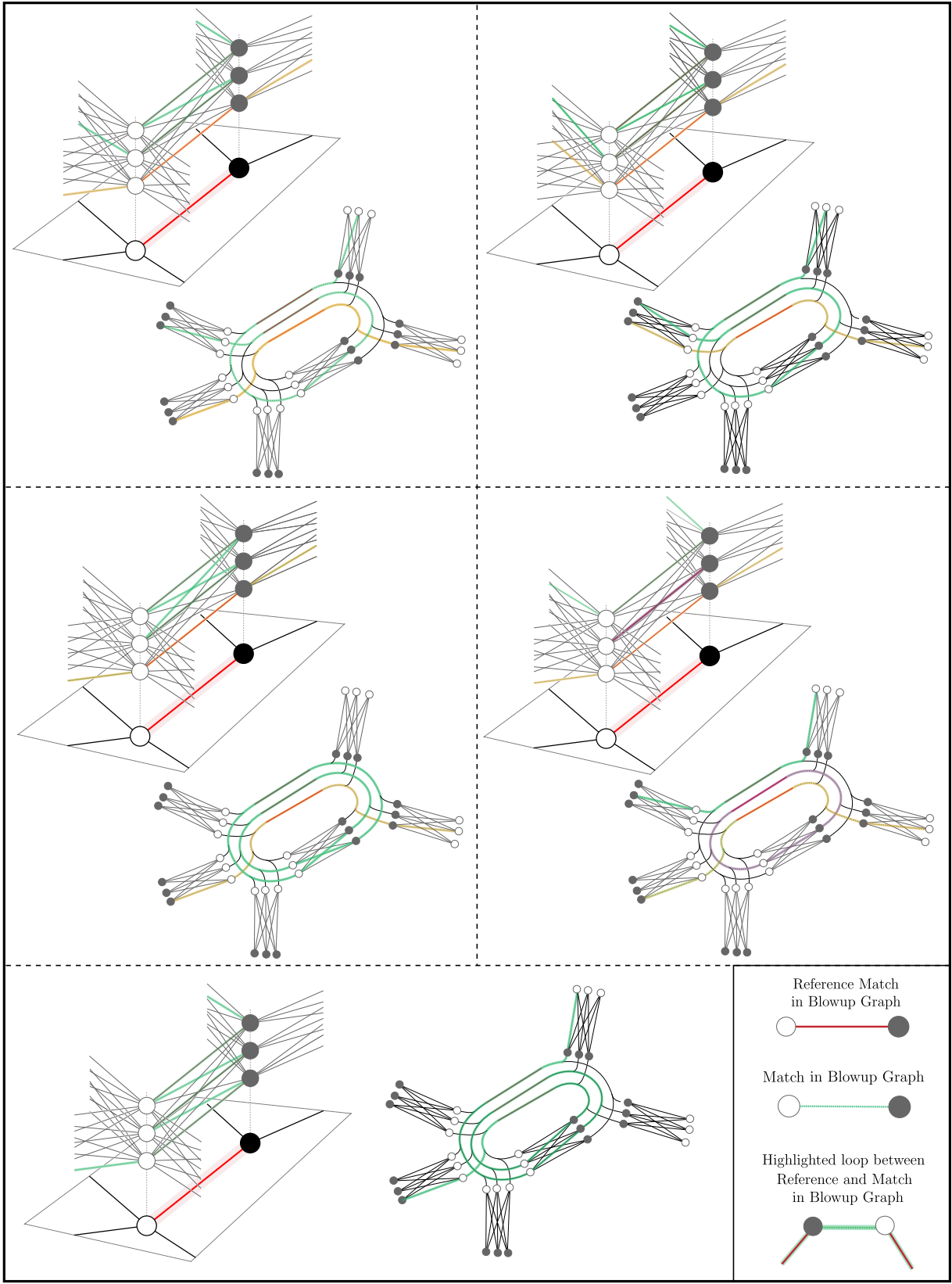}
  \caption{
    Examples of loops between the reference matchings and a matching of the blowup graph. First, the loops in the blowup graph are highlighted different colors for different local attachments of loops. Next, the projections onto the racetrack on the plane are depicted with the same highlights. Some of the loops intersect each other and self-intersect on the racetrack.
  }
  \label{fig:racetrack_and_examples}
\end{figure}

\begin{enumerate}
    \item For each vertex of degree $m$, split the vertex into $m$ vertices so that the edges are drawn independently and disconnected from each other.
    \item Draw each (split up) edge as the complete bipartite graph $K_{N,N}$ which entails splitting each (new, split up) vertex into $N$ more vertices. At this point, a vertex of degree $m$ in the original graph gets split into $m \cdot N$ vertices. This complete, bipartite graph will be drawn parallel to the edge: vertices will be split up transversely to the edge and the edges of the bipartite graph will be drawn in the same direction as the edge.
    \begin{enumerate}
        \item If $b,w$ are the vertices of the original graph that the edge corresponds to, then the labels of the vertices of this new graph will be mapped onto vertices of the blowup graph. So, each $v^{(1)},\cdots,v^{(N)}$ for $v$ of degree $m$ will correspond to $m$ vertices of the racetrack graph. 
        \item Relative to the edge $b \to w$ (black-to-white), this is arranged so that the left-to-right order of the vertices corresponds to $b^{(1)},\cdots,b^{(N)}$ and $w^{(1)},\cdots,w^{(N)}$. As such, the blowup edge $(b^{(1)}, w^{(1)})$ is the leftmost edge and $(b^{(N)}, w^{(N)})$ is the rightmost, and the rest of the edges are drawn in between as in Fig.~\ref{fig:racetrack}.
    \end{enumerate}
    \item Next to each reference edge, there's a `racetrack' which connects all of the disjoint bipartite graphs above. Let $(b,w)$ be the reference edge in question. Say $b$ has multiplicity $m$ and $w$ has multiplicity $m'$. 
    \begin{enumerate}
        \item Refer to the $m \cdot N$ vertices of the racetrack graph corresponding to $b$ as $b^{(i)}_{(j)}$ for $i = 1 \cdots N$ and $j = 1 \cdots m$. Similarly, call the $m' \cdot N$ vertices corresponding to $w$ $w^{(i)}_{(j)}$ for $i = 1 \cdots N$ and $j = 1 \cdots m'$. The labelings are fixed as follows. Each $v^{(i)}_{\cdots}$ will correspond to the vertex $v^{(i)}$ of the blowup graph. 
        The subscript label `$(j)$' will be given to the vertex corresponding to the labeling of the edges in the ciliation with respect to $v$. In particular $b^{(\cdots)}_{(1)}, \cdots, b^{(\cdots)}_{(m)}$ would be the clockwise order at $b$ with $b^{(\cdots)}_{(m)}$ corresponding to the reference edge, and $w^{(\cdots)}_{(1)}, \cdots, w^{(\cdots)}_{(m')}$ would be the counterclockwise order at $w$ with $w^{(\cdots)}_{(m')}$ at the reference edge. 
        \item Draw $N$ edges connecting $b_{(m)}^{(i)}$ to $w_{(m')}^{(i)}$. In particular, draw them in a nested `oval' shapes on the right-side of $b \to w$ in a way avoiding intersections with the other edges. The nesting would mean the outside-to-inside order would be $b_{(m)}^{(1)} \to w_{(m')}^{(1)}, \cdots, b_{(m)}^{(N)} \to w_{(m')}^{(N)}$. These edges together with the original $b_{(m)}^{(i)} \to w_{(m')}^{(i)}$ edges form full concentric ovals, in the shape of a racetrack. Let us temporarily refer to these as $\mathrm{oval}_1, \cdots, \mathrm{oval}_m$
        \item Now, draw `on-ramps' to the racetrack as follows. For each group of white vertices $w^{(1)}_{(j)}, \cdots, w^{(N)}_{(j)}$, draw $N$ parallel edges that connect to the ovals so that $w^{(i)}_{(j)}$ is connected to $\mathrm{oval}_i$. Note that this will introduce a crossing between this new connection and $\mathrm{oval}_{i'}$ for $i' < i$. In Fig.~\ref{fig:racetrack}, we depict these as overcrossings where the oval is on top. Do the same for groups of black vertices $b^{(1)}_{(j)}, \cdots, b^{(N)}_{(j)}$.
        \item After repeating this procedure for every reference edge, the racetrack graph is completed.
    \end{enumerate}
    \item Now, we describe how to take a matching $\mathcal{M}$ of the blow-up graph and draw the loops of $\mathcal{M} \sqcup \tilde{\mathcal{R}}$ on the racetrack network.
    \begin{enumerate}
        \item At this point, we notice that the region of the ovals clockwise to the $w^{(\cdots)}_{(1)}$ on-ramps and counter-clockwise to the $b^{(\cdots)}_{(1)}$ on-ramps consists of $N$ parallel edges uninterrupted by additional on-ramps. We interpret these as the projection of the edges in the reference matching with which we'll compose the matching $\mathcal{M}$ to form loops.
        \item Given $\mathcal{M}$, one can first naturally draw these edges in the blowup graph in the complete bipartite graphs embedded in the plane. In particular, the blowup of $(b',w')$ is naturally identified with edges between the vertices $\{(b')^{(\cdots)}_{(j_1)}\}$ and $\{(w')^{(\cdots)}_{(j_2)}\}$, where $j_1$ (resp. $j_2$) is the labeling of the edge in the ciliation of $b'$ (resp. $w'$).
        \item Now, continue the edges via the on-ramps into the race-track. From the white $w^{(\cdots)}_{(\cdots)}$, continue them into the reference edges going clockwise. And from the black $b^{(\cdots)}_{(\cdots)}$, continue them into the reference edges going counterclockwise. Alternatively, one can think of the on-ramps coming from $b^{(\cdots)}_{(\cdots)}$ as `off-ramps' and continue the edges going clockwise around the racetrack, starting from the white on-ramps and off onto the black off-ramps.
        \item Because each $w^{(i)}$, $b^{(i)}$ in the blowup graph is part of exactly one edge in $\mathcal{M}$, every $\mathrm{oval}_i$ will have exactly one edge continued in from the white side and one edge continued out of the black side. In particular, all crossings between loops will be in general position, either between edges in the bipartite graphs or coming from crossings of the on-ramps with the ovals. We describe the crossings now in more detail. 
        \begin{enumerate}
            \item First, note that crossings in the racetrack graph corresponding to edge $e'=(b',w')$ are in correspondence with inversions in the permutation $\sigma_{e'}$, since the edge $(b')^{(i_1)}_{(\cdots)} \to (w')^{(i'_1)}_{(\cdots)}$ crosses $(b')^{(i_2)}_{(\cdots)} \to (w')^{(i'_2)}_{(\cdots)}$ iff (say WLOG $i_1 < i_2$) $i'_1 > i'_2$.
            \item Next, note that crossings corresponding to intersections of the on-ramps/off-ramps with the ovals correspond to inversions in $\sigma_v$. Recall that every $\mathrm{oval}_i$ has a part that is continued from exactly one of the $\{w^{(i)}_{(j)}\}_j$ and back to exactly one of the $\{b^{(i)}_{(j')}\}_{j'}$ and that each of these continuations crosses all of the $\mathrm{oval}_{i'}$ for $i' < i$. But, not all of these continued loops will cross a loop. For $j_1 < j_2$, we'll have that the loop coming from $w^{(i_1)}_{(j_1)}$ (resp. $b^{(i_1)}_{(j_1)}$) will cross that coming from $w^{(i_2)}_{(j_2)}$ (resp. $b^{(i_1)}_{(j_1)}$) iff $i_1 > i_2$, since $w^{(\cdots)}_{(j_1)}$ (resp. $b^{(\cdots)}_{(j_1)}$) lie clockwise (resp. counterclockwise) from $w^{(\cdots)}_{(j_2)}$ (resp. $b^{(\cdots)}_{(j_2)}$). 
            By definition, these correspond exactly to the inversions of $\sigma_v$.
            \item \label{item:blowup_loop_projection} As such, the total number of crossings of loops corresponding to $\mathcal{M}$ in the racetrack graph is $\sum_e \ell(\sigma_e) + \sum_v \ell(\sigma_v)$, where $\ell(\sigma)$ is the length (i.e. number of inversions) of the permutation $\sigma$.
        \end{enumerate}
    \end{enumerate}
\end{enumerate}

\subsection{Proof Of Theorem~\ref{thm:dks_on_surfaces}} \label{sec:proof_main_thm}

We first note that expanding the determinant of $K_\eta(\Phi)$ gives
\begin{equation}
\begin{split}
    \det(K_\eta(\Phi)) 
    &= 
    \sum_{\sigma \in \mathrm{Sym}(n \cdot N)} (-1)^{\sigma} (K_\eta(\Phi))_{{\tilde{b}}_1 {\tilde{w}}_{\sigma(1)}} \cdots (K_\eta(\Phi))_{{\tilde{b}}_{n \cdot N} {\tilde{b}}_{\sigma(n \cdot N)}}
    \\
    &=
    \sum_{
        \substack{
            \text{matchings } \mathcal{M} \\
            \text{of blow-up graph}
        }
    }
    (-1)^{\sigma} \epsilon_{\tilde{b}_1, \tilde{w}_{\sigma(1)}} \cdots \epsilon_{\tilde{b}_{n \cdot N}, \tilde{w}_{\sigma(n \cdot N)}}
    \prod_{ (b^{(i)}, w^{(j)} ) \in \mathcal{M}}
    (\phi_{b, w})_{i,j}
    .
\end{split}
\end{equation}
In the above, a base graph with $n$ vertices yields a blowup graph with $n \cdot N$ vertices. The $\tilde{b}_1, \cdots, \tilde{b}_{n \cdot N}$ $\tilde{w}_1, \cdots, \tilde{w}_{n \cdot N}$ denote a labeling of the blowup graph. 

As before, the fact that $K_\eta(\Phi)(b^{(i)},w^{(j)}) = 0$ whenever $(b,w)$ isn't an edge automatically organizes the determinant into a sum over matchings of the blowup graph, where a $(b^{(i)},w^{(j)})$ in the matching would correspond to one of the $(\tilde{b}_k,\tilde{w}_{\sigma(k)})$. Organizing the sums over blowup matchings into webs is done by collecting all terms with the same number of blowup edges in the matching above each edge. 

As such, the remaining issue is to show that each of the signs is coherent with the claim of Thm.~\ref{thm:dks_on_surfaces}, i.e. that
\begin{equation}
    (-1)^{\sigma} \epsilon_{\tilde{b}_1, \tilde{w}_{\sigma(1)}} \cdots \epsilon_{\tilde{b}_{n \cdot N}, \tilde{w}_{\sigma(n \cdot N)}}
    =
    \pm (-1)^s
\end{equation}
for a global $\pm$ sign.

The logic for this is completely analogous to the Rank-1 cases in Sec.~\ref{sec:rank_1_dimer}. Namely, we will be given Grassmann variables $\{ \psi_{w^{(j)}}, \overline{\psi}_{b^{(i)}} | b, w \,\, \text{black, white vertices} \,\, i,j = 1,\cdots,N \}$ corresponding to the vertices of the blowup graph.
Then, we want to show that for a matching $\mathcal{M}$ of the blowup graph, there is a uniform $\pm$ sign over all matchings so that
\begin{equation} \label{eq:higher_rank_grassmann_sign_claim}
    \prod_{(b^{(i)}, w^{(j)}) \in \mathcal{M}} 
    (-\overline{\psi}_{b^{(i)}} \epsilon_{b,w}    
    \psi_{w^{(j)}})
    =
    \pm (-1)^s \prod_{i=1}^{n} \psi_{b_{i}^{(1)}} \overline{\psi}_{w_{i}^{(1)}} \cdots \psi_{b_{i}^{(N)}} \overline{\psi}_{w_{i}^{(N)}} 
\end{equation}

First, we will note that if we choose the matching $\mathcal{M}$ to be the lift $\tilde{\mathcal{R}}$ of the reference matching associated to the cilia as in Fig.~\ref{fig:blowup_graph_and_match} and the racetrack construction, we'll have $(-1)^s = 1$, since that would mean each of the factors $(-1)^{\sigma_e} = (-1)^{\sigma_v} = (-1)^{q_\eta(\mathcal{M} \sqcup \mathcal{R})} = 1$. 

It turns out that Lemma~\ref{lem:grassmann_equals_winding} also applies to matchings on the blowup graph if we project loops onto the surface via the racetrack construction. We reformulate and show the lemma here. 
\begin{lem} \label{lem:grassmann_equals_winding_blowup}
    Let ${\bf e}_1 \to {\bf e}_2 \to \cdots \to {\bf e}_{2m-1} \to {\bf e}_{2m} \to {\bf e}_1$ be a loop of edges in the blowup graph as follows.
    First, any vertex appears at most once. 
    Next, the odd edges ${\bf e}_1,{\bf e}_3,\cdots$ should be taken to be in the blowup reference matching $\tilde{\mathcal{R}}$. 
    And, the even edges ${\bf e}_2,{\bf e}_4,\cdots$ should be \textit{away from} the blowup reference matching as part of the matching. 
    Denote ${\bf e}_{2k+1} = (w^{(i_{k})}_{\ell_k},b^{(i_{k+1})}_{\ell_{k+1}})$ as edges of a matching and ${\bf e}_{2k} = (b^{(i_k)}_{\ell_k},w^{(i_k)}_{\ell_{k}})$ as edges of the reference matching, with the first and last vertices in the loop identified. Then,
    \begin{equation}
        \left(
            \prod_{k=1}^{m}
            \left(-\epsilon_{{\bf e}_{2k}}   \overline{\psi}_{b^{(i_k)}_k} \psi_{w^{(i_k)}_{k+1}} \right)
            \left( \epsilon_{{\bf e}_{2k+1}} \overline{\psi}_{w^{(i_{k})}_k} \psi_{b^{(i_{k+1})}_{k+1}}     \right)
        \right)
        =
        (-1)^{1+\mathrm{wind(loop)}}
        \psi_{b_{\ell_1}} \overline{\psi}_{b_{\ell_1}}
        \psi_{w_{\ell_1}} \overline{\psi}_{w_{\ell_1}}
        \cdots
        \psi_{b_{\ell_1}} \overline{\psi}_{b_{\ell_1}}
        \psi_{w_{\ell_1}} \overline{\psi}_{w_{\ell_1}}
    \end{equation}
    where $\mathrm{wind(loop)}$ is determined by the vector field of Sec.~\ref{sec:geometric_setup} and the projection of the loop onto the surface is given as in the racetrack construction of Sec.~\ref{sec:racetrack_construction} and Figs.~\ref{fig:racetrack},\ref{fig:racetrack_and_examples}.
\end{lem}
\begin{proof}
The proof follows from essentially the same logic as Lemma~\ref{lem:grassmann_equals_winding}.
For the first case of there being no edges in the matching that live on the blowup of a reference edge, the proof is exactly the same, since the computation between windings and Grassmann signs match locally going around the loop. 

A next case of doubled edges consisting of ${\bf e}_1 = {\bf e}_2 = (b^{(i)}, w^{(i)})$ over ${\bf e} = (b,w)$ have Grassmann contribution
\begin{equation*}
	\left(-\epsilon_{{\bf e}}   \overline{\psi}_{b^{(i)}} \psi_{w^{(i)}} \right)
	\left(\epsilon_{{\bf e}}   \overline{\psi}_{w^{(i)}} \psi_{b^{(i)}} \right)
	=
	+\psi_{b^{(i)}} \overline{\psi}_{b^{(i)}} \psi_{w^{(i)}}  \overline{\psi}_{w^{(i)}} 
\end{equation*}
which matches with the winding of the small non-self-intersecting loop this pair would project to, which is one loop around $\mathrm{oval}_i$ of the racetrack.

The most general case follows from considering what happens when there's consecutive edges on the same blow up edge
$$
{\bf e}_k        = ( (b')^{(i')}, w^{(i_1)} ) \, , \,
{\bf e}_{k+1}    = ( w^{(i_1)}, b^{(i_1)} ) \, , \,
{\bf e}_{k+2}    = ( b^{(i_1)}, w^{(i_2)} ) \, , \,
\cdots \, , \,
{\bf e}_{k+2r-1} = ( w^{(i_r)}, b^{(i_r)} ) \, , \,
{\bf e}_{k+2r}   = ( b^{(i_r)}, (w')^{(j')} )
$$
with $i_1, \cdots, i_r$ distinct and possibly $b' \neq b, w' \neq  w$. In particular, each ${\bf e}_{k+2r-1} = (w^{(i_k)},b^{(i_k)})$ is part of the reference matching and each of the $\{ (b^{(i_k)},w^{(i_{k+1})}) \}$ are not reference edges. This corresponds to a section of the loop with $2r-1$ consecutive edges 
$$
{\bf e}_{k+1}    = ( w^{(i_1)}, b^{(i_1)} ) \, , \,
{\bf e}_{k+2}    = ( b^{(i_1)}, w^{(i_2)} ) \, , \,
\cdots \, , \,
{\bf e}_{k+2r-1} = ( w^{(i_r)}, b^{(i_r)} )
$$
in the matching in a row being above the original reference edge.

Then, locally the Grassmann signs reorganize as  
\begin{equation*}
\begin{split}
	&\left(-\epsilon_{(b', w)}   \overline{\psi}_{(b')^{(i')}} \psi_{w^{(i_1)}} \right)
	\left(\epsilon_{{\bf e}}   \overline{\psi}_{w^{(i_1)}} \psi_{b^{(i_1)}} \right)
	\left(-\epsilon_{{\bf e}}   \overline{\psi}_{b^{(i_1)}} \psi_{w^{(i_2)}} \right)
	\cdots
	\left(\epsilon_{{\bf e}}   \overline{\psi}_{w^{(i_r)}} \psi_{b^{(i_r)}} \right)
	\left(-\epsilon_{(b,w')}   \overline{\psi}_{b^{(i_r)}} \psi_{(w')^{(j')}} \right)
	\\
	&=
	(-1)^{r-1} \epsilon_{(b', w)} \epsilon_{(b,w')}
	\overline{\psi}_{(b')^{(i')}} \psi_{w^{(i_1)}}
	\overline{\psi}_{w^{(i_1)}} \psi_{b^{(i_1)}}
	\overline{\psi}_{b^{(i_1)}} \psi_{w^{(i_2)}}
	\cdots
	\overline{\psi}_{w^{(i_r)}} \psi_{b^{(i_r)}}
	\overline{\psi}_{b^{(i_r)}} \psi_{(w')^{(j')}}.
\end{split}
\end{equation*}
This means that locally, the effect of going back-and-forth on the reference edge $r-1$ times gives an additional Grassmann sign of $(-1)^{r-1}$, i.e. $ (-1)^{r-1} \epsilon_{(b', w)} \epsilon_{(w, b)} \epsilon_{(b,w')} \cdots $ as opposed to $ \epsilon_{(b', w)} \epsilon_{(w, b)} \epsilon_{(b,w')} \cdots $. This matches the winding, because going back and forth in the blowup graph like this corresponds to winding around the racetrack ${r-1}$ times, which gives an additional $(-1)^{r-1}$ contribution to the winding factor of the loop. 

Now we can compare windings by iteratively removing extra edges above the reference matching (i.e. induction on $r$) from the reference. This reduction changes the $(-1)^{\mathrm{wind}}$ and the Grassmann sign both by $-1$.
The two base-cases we're left with will be a loop corresponding to a matching with no edges above the original reference, or a doubled-edge, corresponding to a matching with a single edge in the blowup reference. We've showed that for these, the winding and Grassmann signs are the same. So by induction we have that all winding and Grassmann signs are the same.
\end{proof}

Now, there is one more step. By this point, we have that the difference betweeen the sign of the reference matching and the matching of the blowup graph is $\prod_{\mathrm{loops}}(-1)^{1 + \mathrm{wind(loop)}}$. The next claim would give a proof of this web-trace formula.
\begin{lem}
	\begin{equation}
		(-1)^s = \prod_{\mathrm{loops}}(-1)^{1 + \mathrm{wind(loop)}}
	\end{equation}
\end{lem}

\begin{proof}
    Recall we defined $(-1)^s = (-1)^{q_\eta(\mathcal{M} \sqcup \tilde{\mathcal{R}})} \prod_{v} (-1)^{\sigma_v} \prod_{e} (-1)^{\sigma_e}$. 
	
    Note that by item~\ref{item:blowup_loop_projection} in the discussion of the racetracks, the factor $\prod_{v} (-1)^{\sigma_v} \prod_{e} (-1)^{\sigma_e}$ of the sign is the total number of intersections and self-intersections of the loops between the matching and the reference in the racetrack graph in the projection of the blowup graph. 
    
    We use the fact (Lemma~\ref{lem:resolving_intersection}) that resolving an intersection changes the product $\prod_{\mathrm{loops}}(-1)^{1 + \mathrm{wind(loop)}}$ by $-1$. In proving that statement, we had chosen resolutions that change the number of loops by $\pm 1$ and left the total winding unaffected. 
    
    So, if we resolve each crossing, we'll get the equality that 
    \begin{equation}
        \prod_{\substack{\mathrm{original} \\ \mathrm{loops}}} (-1)^{1 + \mathrm{wind(loop)}}
        =
        \underbrace{
            \prod_{v} (-1)^{\sigma_v} \prod_{e} (-1)^{\sigma_e}
        }_{\substack{\text{number of intersections} \\ \text{on racetrack graph}}}
        \prod_{\substack{\mathrm{resolved} \\ \mathrm{loops}}}
        (-1)^{1 + \mathrm{wind(loop)}}
    \end{equation}
    But by construction, the now resolved loops have no intersections or self-intersections and lie in the same homology class as $\mathcal{M} \sqcup \tilde{\mathcal{R}}$. This means that
    \begin{equation}
        \prod_{\substack{\mathrm{resolved} \\ \mathrm{loops}}} (-1)^{1 + \mathrm{wind(loop)}}
        =
        (-1)^{q_\eta(\mathcal{M} \sqcup \tilde{\mathcal{R}})} 
    \end{equation}
    by the definition of $q_\eta(\cdots)$ (see Sec.~\ref{app:quad_form_arf}).
    So in total, we have
    \begin{equation}
        \prod_{\substack{\mathrm{original} \\ \mathrm{loops}}} (-1)^{1 + \mathrm{wind(loop)}}
        =
        \prod_{v} (-1)^{\sigma_v} \prod_{e} (-1)^{\sigma_e} \cdot (-1)^{q_\eta(\mathcal{M} \sqcup \tilde{\mathcal{R}})} 
        =
        (-1)^s
    \end{equation}
    which is what we claim.
\end{proof}

Note that Lemma~\ref{lem:grassmann_equals_winding_blowup} implies Eq.~\eqref{eq:higher_rank_grassmann_sign_claim} in the same way as the rank-1 case of Sec.~\ref{sec:rank_1_dimer}. This implies the theorem, so it is proved. 

\section{Discussion and Questions} \label{sec:discussion}
We formulated and proved a version of the web trace formula of~\cite{dksWebs2022} on higher genus surfaces, while elucidating the proof of the planar case. The main points were that the signs appearing in the Kasteleyn determinant can be interpreted in terms of the spin geometry interfacing with loops in the graph. The sign $\prod_{e}(-1)^{\sigma_e} \prod_{v} (-1)^{\sigma_v}$ that appeared in both the web trace and determinant, corresponding to a blowup matching, is interpreted as a factor of $-1$ for resolving each crossing of the loops of the blowup matching with a reference dimers drawn in the racetrack construction, while factors of $(-1)^{q_\eta(\xi)}$ come from the geometry of the resolved, nonintersecting loops.

We expect the analogous result for nonorientable surfaces to hold, which is illustrated for the rank-1 case in~\cite{cimasoniNonorientable2009}. The proof should go through exactly the same, except one is dealing with $\mathrm{Pin}^-$ structures rather than $\Spin$. Any factors of $\Z_2 = \{\pm 1\}$ coming from $(-1)^{1 + \mathrm{wind(loop)}}$ and $(-1)^{q_\eta(\xi)}$ will instead be replaced with analogous $\Z_4 = \{\pm 1, \pm i\}$ factors satisfying a quadratic property, and the $\mathrm{Arf}(\eta) \in \Z_2 \{\pm 1\}$ invariant gets replaced with the Arf-Brown-Kervaire invariant $\mathrm{ABK}(\eta) \in \Z_8 = \{e^{k \cdot \frac{2 \pi i}{8} }| k = 0 \cdots 7 \}$. The fact one adds a factor of $-1$ for each crossing will stay the same. Although, one may need to take care of implementing the racetrack construction and vector fields (in the non-orientable case the vector field gets replaced with a framing of $T\Sigma \oplus \det(T\Sigma)$) in the presence of orientation-reversing walls.

A key technical point of our discussion was the racetrack construction of Sec.~\ref{sec:racetrack_construction} that allowed us to embed in the plane loops between the matchings of the blowup graph and the reference dimer cover. Crucially, this construction relied heavily on a choice of reference dimer covering and all of our results were in terms of the positive ciliation coming from the reference. It would be interesting to see if a similar construction could be meaningful or applicable for a general cililation that doesn't come from a reference matching.

In a similar vein, our methods do not directly give a proof of the analogous results of~\cite{kenyon2023higherrankdimermodels} who study `higher-rank' dimer models where multiplicities of vertices in the blowup graph are not a fixed $N$ but rather vary between vertices. We note that our racetrack construction bears a resemblance to the `gadget' construction of~\cite{kenyon2023higherrankdimermodels}, so we are curious if they are related. However, our construction relies on reference matchings whereas theirs do not so any connection isn't immediately obvious. 

In~\cite{kenyon2023planar3websboundarymeasurement}, certain topological     `connection probabilities' of $\SL_3$ webs were computed in the scaling limit of the upper-plane, generalizing $\SL_2$ results of~\cite{kenyonWilson2011}. And in~\cite{Dubedat2018doubleDimersConformalLoopEnsemblesIsomonodromicDeformations,tata20232dfermionsstatisticalmechanics}, scaling limits of the planar Kasteleyn determinants with respect to connections flat away from a set of punctures were studied and related to `isomonodromic tau functions'. In general, it could be interesting to formulate scaling limit results of similar quantities like topological probabilities on higher genus surfaces or flat connections on punctured surfaces, most naturally the torus. For example, isomonodromic tau functions on the torus have been studied intensely. 

In these veins, it is expected that the dimer model in the scaling limit is the Dirac fermion. As such, studying the Conformal Field Theoretic properties of the scaling limit of these higher-rank dimer models is desirable. Results about CFT interpretations of web connection probabilities are given in~\cite{lafay2024degenerateconformalblocksw3} for $\SL_3$ and~\cite{peltolaWu2019} for $\SL_2$ and are related to the $W_3$ algebra at $c=2$ and $W_2 = \mathrm{Virasoro}$ algebra at $c = 1$. In general, it is expected that the rank $N$ dimer model is related to $W_N$ algebras at $c=N-1$~\cite{lafayLeWN}. This would fit in with the paradigm of the dimer model scaling to the Dirac fermion because the $W_N$ algebra at $c=N-1$ is known to have a natural construction in terms of $N$ free fermions (see e.g.~\cite{gavrMarshFreeFermionWN2016}), which in our interpretation corresponds to a rank-$N$ dimer model.

Higher genus surfaces are important in CFT, since cutting and gluing give important information about a theory's spectrum and topological properties (c.f.~\cite{bigYellowCFT1997}). In particular, partition functions on the torus are fundamental due their simplicity in computation and important information they give. Moreover, the Dirac fermion on higher genus surfaces is known to be related to the scaling limit of the dimer model and interplays with the various spin structures, see~\cite{dijkgraaf2009} for an exposition of the torus partition function and also~\cite{basok2024dimersriemannsurfacescompactified} for recent results about the higher-genus case. As such, the results of this paper can hopefully be used to study scaling limits of twisted Kastleyn determinants and web combinatorics on the torus and compare them to CFT quantities relating rank-$N$ dimers and rank-$N$ Dirac fermions.   

\section*{Acknowledgments}
We acknowledge the Yale Math Department for graduate student funding. We also thank Dan Douglas, Rick Kenyon, Nick Ovenhouse, and Sam Panitch for related discussions that motivated us to formulate the main results and write this paper. 

\appendix

\section{Grassmann Variables} \label{app:grassmann}

Grassmann integration is a compact way to encode antisymmetric structures, such as determinants. Let us now review it here.

The \textit{Grassmann Algebra} over the set $\{1,\cdots,n\}$ is generated by a set of anticommuting variables $\psi_1, \psi_2, \cdots, \psi_n$ over $\C$. This means
\begin{equation}
    \{ \psi_i, \psi_j \}
    :=
    \psi_i \psi_j + \psi_j \psi_i = 0 \quad \text{ for all } i,j,
    \quad\quad\quad\quad \text{ so } \quad\quad\quad\quad
    \psi_i^2 = 0 \quad \text{ for all } i.
\end{equation}
A general \textit{Grassmann function} in these variables can be represented by a polynomial with terms of order $\le 1$ in each variable, like
\begin{equation}
    f(\psi_1,\cdots,\psi_n) = \sum_{k=0}^{n} \sum_{1 \le i_1 < \cdots < i_k \le n} c_{i_1,\cdots,i_k} \psi_{i_1} \cdots \psi_{i_k}.
\end{equation}
A function is \textit{Grassmann even (resp. odd)} if all monomials consist of an even (resp. odd) number of variables. Grassmann even functions commute with all other variables, and odd functions anticommute with other odd functions.

For example, a general function for $n=2$ looks like
\begin{equation}
    f(\psi_1,\psi_2) 
    = 
    c + c_1 \psi_1 + c_2 \psi_2 + c_{12} \psi_1 \psi_2
    =
    c + c_1 \psi_1 + c_2 \psi_2 - c_{12} \psi_2 \psi_1
\end{equation}
Furthermore, introduce more anticommuting variables $d\psi_1,\cdots,d\psi_n$ satisfying
\begin{equation}
    \{ d\psi_i, d\psi_j \}
    =
    \{ d\psi_i, \psi_j \}
    = 0 \quad \text{ for all } i \neq j
\end{equation}
which play the role of differentials for integration. 
Now, consider a nonzero monomial $d\psi_{j_1} \cdots d\psi_{j_\ell} \, \psi_{i_1} \cdots \psi_{i_k}$ in these variables (so that $\{j_1,\cdots,j_\ell\}$ are distinct and $\{i_1,\cdots,i_k\}$ are distinct). For each $i$ there is a pairing $\int_i$ that sends
\begin{equation}
\begin{split}
    &\int_i d\psi_{j_1} \cdots d\psi_{j_{\ell-1}} \underbrace{ d\psi_{i} \, \psi_{i} } \psi_{i_1} \cdots \psi_{i_k} 
    =
    d\psi_{j_1} \cdots d\psi_{j_{\ell-1}} \, \psi_{i_1} \cdots \psi_{i_k}
    \\
    &\int_i d\psi_{j_1} \cdots d\psi_{j_\ell} \, \psi_{i_1} \cdots \psi_{i_k} = 0
    \quad \text{ if } i \not\in \{i_1,\cdots,i_k\} \text{ or } i \not\in \{j_1,\cdots,j_{\ell}\}.
\end{split}
\end{equation}
By using antisymmetry to put variables in the right order, these integrals uniquely determine the integral acting on any monomial, thus on any Grassmann function. Henceforth, we will denote the integral $\int d\psi_{j_1} \cdots d\psi_{j_{\ell}} f$
to refer to the integral
$\int_{j_1} \cdots \int_{j_{\ell}} d\psi_{j_1} \cdots  d\psi_{j_{\ell}} f$ so that the indices of the $d\psi_{\cdots}$ variables determine which variables get integrated over. As examples, we'll have
\begin{align*}
    &\int d\psi = 0, 
    &
    &\int d\psi_1 d\psi_2 \, \psi_2 \psi_1 = \int d\psi_1 \, \psi_1 = 1,
    &
    &\int d\psi_1 \, \psi_2 \psi_1 = -\psi_2,
    \\
    &\int d\psi \psi = 1 = -\int \psi d\psi,
    &
    &\int d\psi_1 d\psi_2 \, \psi_1 \psi_2 = -\int d\psi_2 d\psi_1 \, \psi_1 \psi_2 = -1,
    &
    &\int d\psi_1 d\psi_2 d\psi_3 d\psi_4 \, \psi_2 \psi_1 = 0.
\end{align*}
Now, let $\psi_1,\cdots,\psi_n$, $\overline{\psi}_1,\cdots,\overline{\psi}_n$ by two independent sets of Grassmann variables that all mutually anticommute and  $d\psi_1,\cdots,d\psi_n$, $d\overline{\psi}_1,\cdots,d\overline{\psi}_n$ be their corresponding differentials that also anticommute with each other and the $\{\psi, \overline{\psi}\}$. Then, we have the following fact. 
\begin{prop} \label{prop:grassmanInt_equals_det}
For any matrix $A \in \mathrm{Mat}_{n \times n}$
\begin{equation}
    \int 
    \left( \prod_{i=1}^{n} d\overline{\psi}_i d{\psi}_i \right) 
    e^{- \sum_{i,j} \overline{\psi}_i A_{ij} \psi_j}
    = 
    \det(A).
\end{equation}
\end{prop}
\begin{proof}
First, note that the order of the differentials $d\overline{\psi}_i d{\psi}_i$ doesn't matter since these are even. 
Next, note that since the terms of $-\sum_{i,j} \overline{\psi}_i A_{ij} \psi_j$ all commute, we can write
\begin{equation}
    e^{- \sum_{i,j} \overline{\psi}_i A_{ij} \psi_j}
    =
    \prod_{i,j} e^{ - \overline{\psi}_i A_{ij} \psi_j }
    =
    \prod_{i,j} \left( 1 - \overline{\psi}_i A_{ij} \psi_j \right),
\end{equation}
where the last equality follows from noting that $(\overline{\psi}_i \psi_j )^2 = 0$.
Now after expanding, we note that the only terms the only monomials contributing to the integral will have exactly one of each $\psi_1,\cdots,\psi_n$ and $\overline{\psi}_1,\cdots,\overline{\psi}_n$. These mean we can write 
\begin{equation}
\begin{split}
    &\int 
    \left( \prod_{i=1}^{n} d\overline{\psi}_i d{\psi}_i \right) 
    e^{- \sum_{i,j} \overline{\psi}_i A_{ij} \psi_j}
    \\
    &= 
    \int 
    \left( \prod_{i=1}^{n} d\overline{\psi}_i d{\psi}_i \right) 
    \sum_{\sigma \in \mathrm{Sym}(n)}
    \left( - \overline{\psi}_1 A_{1 \sigma(1)} \psi_{\sigma(1)} \right) 
    \left( - \overline{\psi}_2 A_{2 \sigma(2)} \psi_{\sigma(2)} \right) 
    \cdots
    \left( - \overline{\psi}_n A_{n \sigma(n)} \psi_{\sigma(n)} \right)
    \\
    &= 
    \sum_{\sigma \in \mathrm{Sym}(n)} 
    \int 
    \left( d\overline{\psi}_n d{\psi}_n \cdots  d\overline{\psi}_1 d{\psi}_1 \right) 
    \left( \psi_{\sigma(1)} \overline{\psi}_1  \cdots \psi_{\sigma(n)} \overline{\psi}_n \right)
    A_{1 \sigma(1)} \cdots A_{n \sigma(n)}
    \\
    &= 
    \sum_{\sigma \in \mathrm{Sym}(n)} (-1)^{\sigma}
    \int 
    \left( d\overline{\psi}_n d{\psi}_n \cdots  d\overline{\psi}_1 d{\psi}_1 \right) 
    \left( \psi_{1} \overline{\psi}_1  \cdots \psi_{n} \overline{\psi}_n \right)
    A_{1 \sigma(1)} \cdots A_{n \sigma(n)}
    \\
    &= 
    \sum_{\sigma \in \mathrm{Sym}(n)} (-1)^{\sigma}
    A_{1 \sigma(1)} \cdots A_{n \sigma(n)}
    = \det(A).
\end{split}
\end{equation}
\end{proof}

\section{Spin Structures} \label{app:spin_structures}
We review the notion of spin structures in this section of the Appendix. The next paragraphs are a general summary of spin structures on manifolds. However, most of the necessary material for surfaces is explained in Sec.~\ref{sec:spin_structures_surfaces} and is mostly self-contained there with some references to the next paragraphs. The reader only interested in surfaces can thus skip to Sec.~\ref{sec:spin_structures_surfaces} and use the next paragraphs as a reference for how the surface language connects to more general notions regarding spin structures.

Recall that $\SO(n)$ is always endowed with a double-cover $\Spin(n)$. Formally, a spin structure on an oriented $n$-manifold $M$ is a lift of the principal $\SO(n)$-bundle $P_{\SO}$ associated to the tangent bundle $TM$ to a principal $\Spin(n)$-bundle $P_{\Spin}$ together with a map $P_{\Spin} \to P_{\SO}$ that restricts to the two-fold covering map $\Spin(n) \to \SO(n)$ on each fiber. Spin structures are well-known to exist iff the second Stiefel-Whitney class of a manifold vanishes in cohomology and to be in one-to-one correspondence with $H^1(M,\Z_2)$. 

An important case will be spin structures on 1-manifolds,  where $\SO(1)$ is the trivial group and $\Spin(1) = \Z_2$. Without loss of generality, we can consider the 1-manifold $S^1$. In this case, there are exactly two principal $\Spin(1) = \Z_2$ bundles, the disconnected trivial bundle $S^1 \times \Z_2$ and the Möbius bundle. We often refer to the bundle $S^1 \times \Z_2$ as the `periodic spin structure' and the Möbius bundle as the `antiperiodic spin structure'. This is because the monodromy of the only possible parallel transport on $S^1 \times \Z_2$ will be $+1$, which can be thought of as `periodic boundary conditions' on the circle. And the monodromy around the Möbius bundle will be $-1$, corresponding to `antiperiodic boundary conditions'. 

A more concrete description like the one used in this paper can be given geometrically. We follow the exposition of~\cite{scorpan2005wild}.  \textit{In short, a spin structure can be thought of as a consistent assignment of (mod 2) winding numbers to framed loops which is invariant under regular homotopy in the manifold.} We will see that this is equivalent to a choice of induced spin structure, or boundary conditions, on each loop. The winding is defined in terms of in terms of a frame of vector fields $v_1, \cdots, v_n$ on the manifold whose orientation $v_1 \wedge \cdots \wedge v_n$ agrees with the manifold everywhere except on a possibly codimension-2 submanifold where the frame becomes singular (i.e. where $v_1 \wedge \cdots v_n = 0$). This frame must vanish in a certain way with respect to the codimension-2 submanifold, which we describe soon. In particular, the relative framing between the framed loop and $v_1, \cdots, v_n$ will give a map $R: \mathrm{loop} \to \GL(n)$. Such an $R$ can be identified with an element of $\pi_1(\GL(n))$, which is $\Z$ for $n=2$ and $\Z_2$ for $n>0$ and will be precisely the definition of the (mod 2) winding number (upon (mod 2) reduction for $n=2$). Now we identify $R$ with an induced spin structure. First note that an element of $\pi_1(\GL(n))$ corresponds exactly to an element of $\pi_1(\SO(n))$ since $\SO(n)$ is a deformation retraction of $\GL(n)$. Then note that any element of $\pi_1(\SO(n))$ lifts to a curve in $\Spin(n)$ starting at $1$ and ending at $\pm 1$. The final position $\pm 1$ will be the induced spin structure, where $+1$ corresponds to the periodic/trivial bundle and $-1$ to the antiperiodic/Möbius bundle.
For this winding number to be regular-homotopy invariant, we will need a condition that this (mod 2) winding stays the same upon a moving across the aforementioned codimension-2 submanifold of the singular framing. This condition can be interpreted as the singularity of $v_1 \wedge \cdots \wedge v_n$ having an `even index'. 

We say two spin structures are equivalent if they assign the same windings to every framed curve. In general, there are many relations between windings of framed curves. For example, twisting the framing by a unit will change the winding by $\pm 1$. And in general, making local changes to the curve will affect the induced spin structure in a way that doesn't depend on the vector field used to define the spin structure. In fact, if one knows these windings for framed curves in some basis of the $\Z_2$ homology of a manifold, then that data determines the windings on all other curves. In addition, one can find vector fields that assign \textit{any} induced spin structure to any basis. These statements are one perspective on the fact that spin structures are in correspondence with $H^1(M,\Z_2)$.

In addition, there is a natural $H^1(M,\Z_2)$ action on spin structures that connects any two spin structures. Let $a \in H^1(M,\Z_2)$ and $\eta$ be a spin structure. Then, the spin structure $a \cdot \eta$ assigns induced spin structures to framed loops as follows. Say that $(-1)^{q_\eta(\mathrm{loop})} = +1$ if the loops is induced with an antiperiodic spin structure by $\eta$ and $(-1)^{q_\eta(\mathrm{loop})} = -1$ if it's induced with a periodic one. \footnote{This sign convention, seemingly at odds with the naming of periodic and antiperiodic, is meant to parallel an important construction for surfaces, detailed in Sec.~\ref{app:quad_form_arf}.}
Then, the action on induced spin structures can be described as 
\begin{equation} \label{eq:twist_spinStruct}
    (-1)^{q_{a \cdot \eta}(\mathrm{loop})} = (-1)^{\braket{a,\mathrm{loop}}} (-1)^{q_\eta(\mathrm{loop})}
\end{equation}
where $\braket{\cdot,\cdot} : H^1(M,\Z_2) \times H_1(M,\Z_2) \to \Z_2$ is the cohomology-homology pairing.
This can be seen in the vector field picture as follows. It is known that every class in $H^1(M,\Z_2)$ is dual to a codimension-1 submanifold. The vector field frame associated to $a\cdot\eta$ is constructed by twisting the frame by $360^\circ$ going across the dual of $a$. This will have exactly the effect of changing the endpoint of the lift of the endpoint $\pm 1$ of the path to $\mathrm{Spin(n)}$ by $-1$ for every time the curve passes the dual for $a$. This precisely corresponds to the Eq.~\eqref{eq:twist_spinStruct}. See Fig.\ref{fig:twisting_vect_field_spin_struct} for an illustration in the $n=2$ case.

\subsection{Spin Structures on Surfaces} \label{sec:spin_structures_surfaces}

On an oriented surface, the picture is even more concrete and simpler. In particular, the orientation of a manifold assigns a canonical framing to every curve, so we don't need to worry about framing curves. In addition, the vector field framing $v_1, v_2$ is essentially determined by $v_1$, since the relevant windings can be computed using $v_1$ and the orientation.

We will see that a spin structure on an oriented surface can be identified with a vector field with all even index singularities in the surface. As we said above, this vector field has the functional property of determining induced spin structures on immersed curves. Specifically for a loop, the winding number $\mathrm{wind(loop)}$ (i.e. the number of full rotations the vector field makes relative to the curve's tangent) of the vector field relative to the tangent of the curve determines this. We have
\begin{equation*}
    \text{induced spin structure}(\mathrm{loop})
    =
    \begin{cases}
        \text{antiperiodic}, &\quad\text{ if } (-1)^{\mathrm{wind(loop)}} = -1
        \\
        \text{periodic}, &\quad\text{ if } (-1)^{\mathrm{wind(loop)}} = +1
    \end{cases}.
\end{equation*}
We can think of this induced spin structure as coming from halving the winding of the vector field. See Fig.~\ref{fig:loops_and_windings}. 

Now it will be useful to explicitly define $\mathrm{wind(loop)}$ and the index of a singularity. The vector field and the loop define a relative angle $\alpha(\theta) \in \R$ for a choice of $\theta \in [0,2\pi]$ parameterizing the curve. Crucially, we have that $\alpha(\theta)$ lives in $\R$ rather than $(0,2\pi)$ because we need it to be a continuous function. The choice of $\alpha(0)$ can be shifted by an arbitrary element of $2 \pi \Z$ so is not canonical. However, $\frac{d \alpha}{d \theta}$ is well-defined without extra choices. We'll have:
\begin{equation} \label{eq:winding_def}
    \mathrm{wind(loop)} = \frac{1}{2\pi} \int_{0}^{2\pi} \frac{d \alpha}{d \theta} d\theta  = \frac{1}{2\pi}(\alpha(2\pi) - \alpha(0)).
\end{equation}
Now, given a singularity of the vector field, the index is defined in terms of `$\mathrm{loop}$' being a small counterclockwise loop around the singularity, namely:
\begin{equation}
    \mathrm{index(singularity)} = \mathrm{wind(loop)} + 1.
\end{equation}

These windings will satisfy several key properties, also depicted in Fig.~\ref{fig:loops_and_windings}.
\begin{lem} \label{lem:winding_across_sing}
    Moving the curve across an index-$k$ singularity changes the winding of the curve by $\pm k$.
\end{lem}

\begin{lem} \label{lem:winding_reversing_loop}
    Reversing a loop negates its winding. I.e.  $\mathrm{wind} ( \mathrm{loop} \, {}^{-1} ) = -\mathrm{wind(loop)}$.
\end{lem}

\begin{lem} \label{lem:resolving_intersection}
     Given a generic intersection of loops, resolve the intersection in a way consistent with the loops' directionality. Then, the number of involved loops changes by $\pm 1$ from $2 \to 1$ if it's an intersection between two different loops or $1 \to 2$ if it's a self-intersection of a loop. And, the winding of the single loop is the sum of the windings of the two loops.
\end{lem}

The first two properties above are straightforward. The third is also simple, but we illustrate it here.
\begin{proof}[Proof of Lem~\ref{lem:resolving_intersection}]
The number of loops changing by $\pm 1$ is a combinatorial statement about loop reconnection, illustrated in Fig.~\ref{fig:loops_and_windings}. The windings being unchanged can be seen in two steps. First, make small isotopies of the loops to make them tangent at the intersection point. This doesn't change the total windings of the loops. Next, reconnect the loops and move them away from each other. This reconnection will keep the total winding the same. In particular, the winding on a curve can be thought is a local integral, Eq.~\eqref{eq:winding_def} along the curve. Reconnecting them while they're tangent simply splices the total integration, reparameterizes the curves, and rearranges the winding integrals.
\end{proof}

At this point, we've justified the fact that a spin structure on a surface is a vector field with even-index singularities. In particular, Lemma~\ref{lem:winding_across_sing} says that the (mod 2) winding of an immersed curve stays invariant under regular homotopy.

The condition that the vector field is even-index means that this lift is regular-homotopy invariant, since moving a curve across an even index singularity keeps the winding number (mod 2) the same, thus keeps the induced spin structure the same. Moving across an odd index singularity twists the induced spin structure on a curve to the opposite one, so wouldn't be allowed. Note that for a small topologically trivial non-self-intersecting curve away from any singularities of the vector field, the vector field looks locally constant, so $\mathrm{wind(loop)} = \pm 1$, which means that $(-1)^{\mathrm{wind(loop)}} = -1$ for a trivial loop.

As before, we say that two spin structures are \textit{equivalent} if they assign equivalent windings to loops. The action of $H^1(M,\Z_2)$ on spin structures is illustrated in Fig.~\ref{fig:twisting_vect_field_spin_struct}.

\afterpage{
\begin{figure}[h!]
  \centering
  \includegraphics[width=\linewidth]{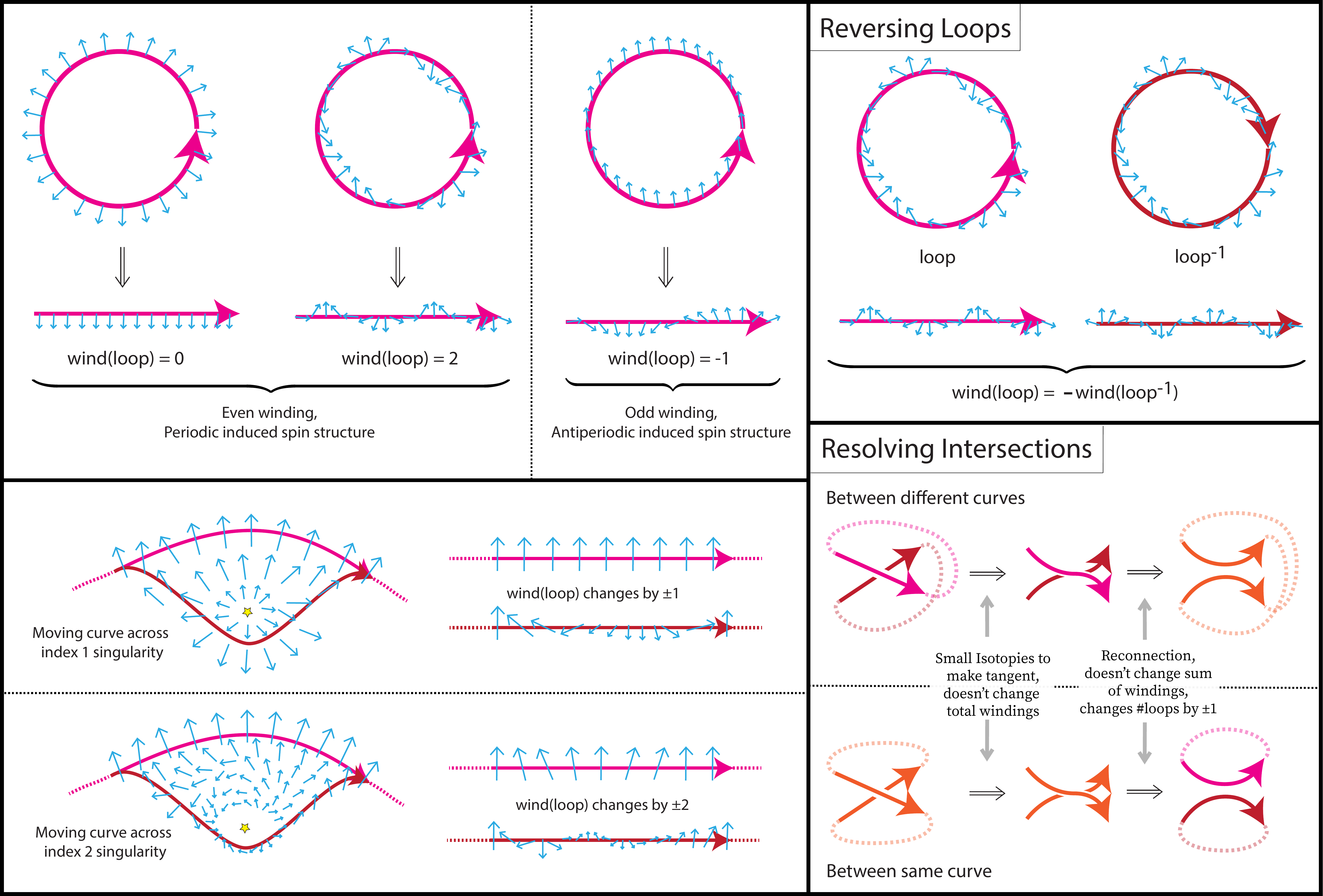}
  \caption{
    (Top-Left) Depiction of induced spin structures on curves coming from the winding of the tangent relative to the vector field. NOTE: we depict here abstract loops on a surface with the vector field relative to the tangent. But if we think of them as bounding a disk in $\R^2$ as depicted, the vector fields can be smoothly continued to the interior of the disk with singularities of index 1,3,0 respectively in left-to-right order.
    (Bottom-Left) Lemma~\ref{lem:winding_across_sing} illustrated for $k=1,2$.
    (Top-Right) Illustration of Lemma~\ref{lem:winding_reversing_loop}.
    (Bottom-Right) Illustration of Lemma~\ref{lem:resolving_intersection}.
  }
  \label{fig:loops_and_windings}
\end{figure}

\begin{figure}[h!]
  \centering
  \includegraphics[width=\linewidth]{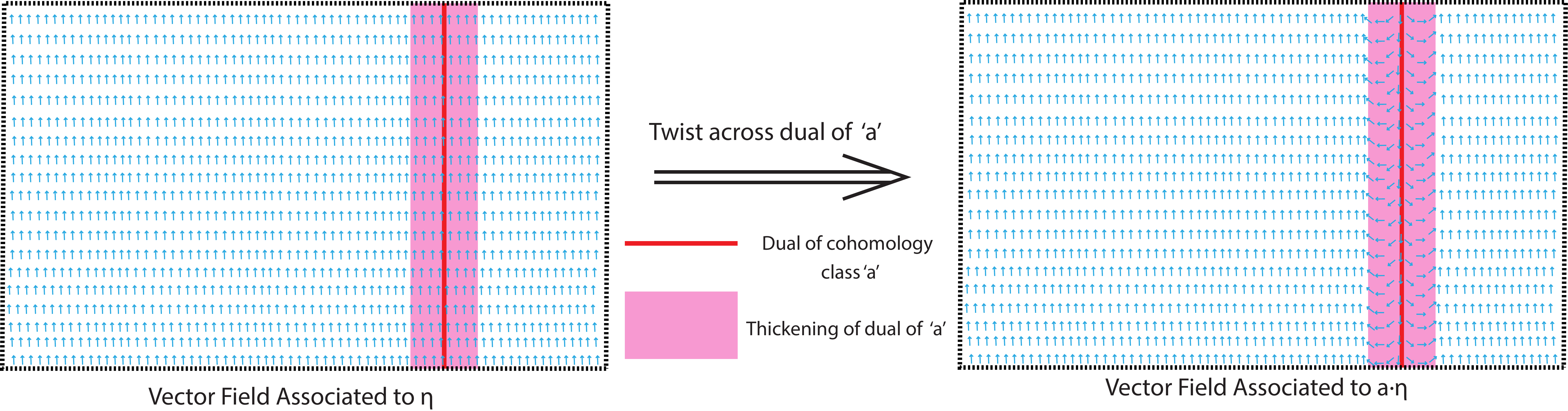}
  \caption{
    Given a spin structure, $\eta$ and a cohomology class $a \in H^1(M,\Z_2)$, the spin structure $a \cdot \eta$ is formed by twisting $\eta$'s associated vector field across the dual of $a$.
    In this picture, we are assuming we are `zoomed in' close enough to the wall so that the vector field looks locally constant. 
  }
  \label{fig:twisting_vect_field_spin_struct}
\end{figure}
\clearpage
}

\afterpage{
\begin{figure}[h!]
  \centering
  \includegraphics[width=\linewidth]{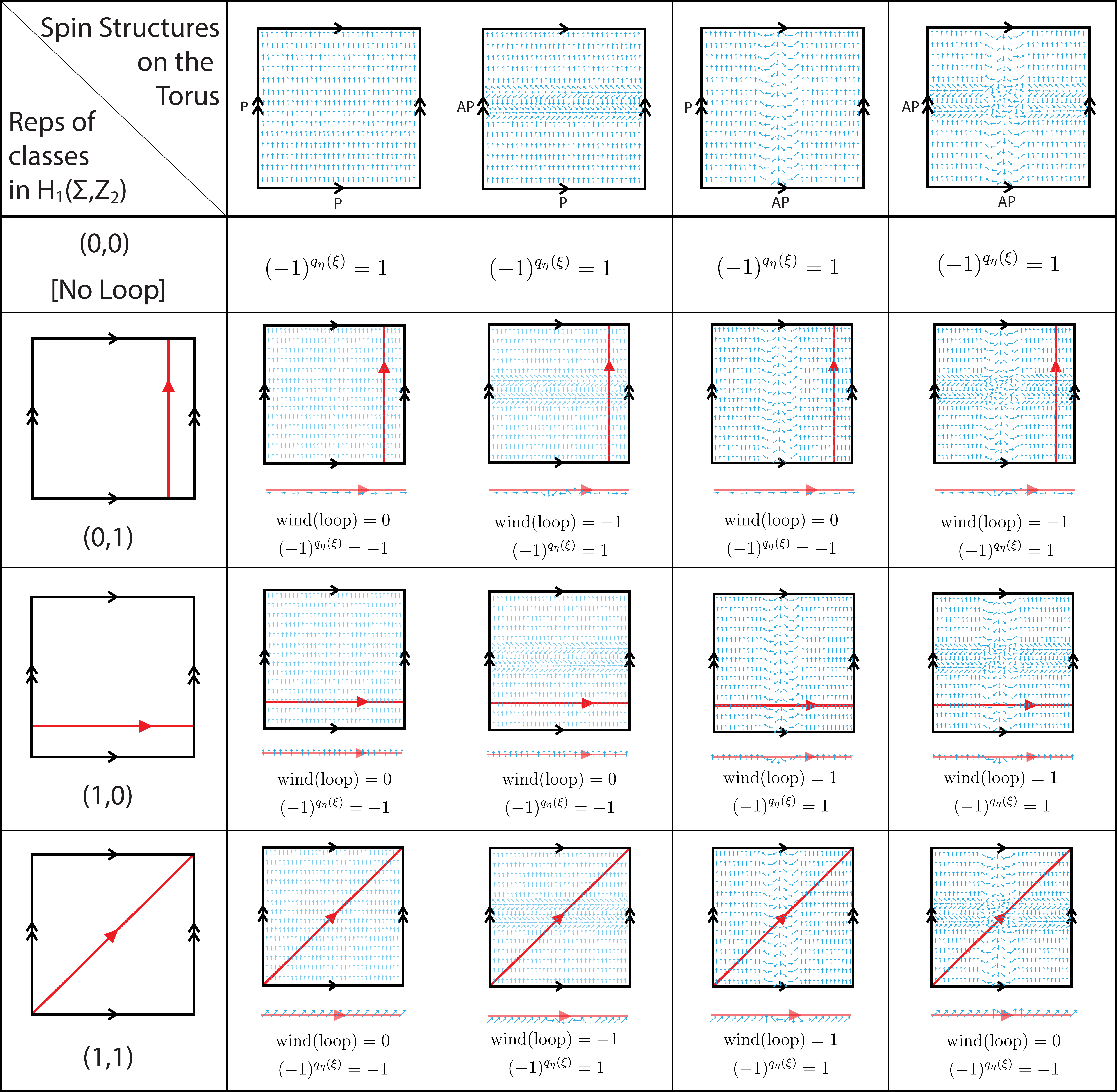}
  \caption{
    Spin structures on a torus together with a depiction the evaluation of $q_\eta(\xi)$ for each of the four homology classes. The labels $P, AP$ around the boundary cycles refers to whether the induced spin structure is periodic or antiperiodic going around the cycle. Note that on a torus, the vector fields can be chosen to be everywhere non-vanishing, but on other surfaces, they will generically need some even-index singularities by the Poincaré-Hopf theorem.
    Note that in the basis $P/P, P/AP, AP/P, AP/AP$ of spin structures and $(0,0),(0,1),(1,0),(1,1)$ of homology, the quadratic form matrix $A_{\eta,\xi} = (-1)^{q_\eta(\xi)}$ and its inverse $(A^{-1})_{\xi,\eta} = \frac{1}{2^{2g}} (-1)^{q_\eta(\xi)}$ can be written as 
    \begin{equation*}
        A 
        =
        \begin{pmatrix}
            1  & 1  & 1  & 1  \\
            -1 & -1 & 1  & 1  \\
            -1 & 1  & -1 & 1  \\
            -1 & 1  & 1  & -1
        \end{pmatrix}
        \quad\quad\text{and}\quad\quad
        A^{-1} 
        = \frac{1}{4}
        \begin{pmatrix}
            1 & -1 & -1 & -1  \\
            1 & -1 & 1  & 1   \\
            1 & 1  & -1 & 1   \\
            1 & 1  & 1  & -1
        \end{pmatrix}.
    \end{equation*}
  }
  \label{fig:spin_structs_torus}
\end{figure}
\clearpage
}

\subsubsection{Quadratic Forms and Arf Invariant} \label{app:quad_form_arf}

It was an observation of~\cite{johnson1980} that given a spin structure, one can construct a $\Z_2$-homology invariant from these windings. In particular, consider a Riemann surface $\Sigma$ equipped with a spin structure $\eta$. Then consider a $\Z_2$ homology class $\xi$ which is represented by a collection of loops $\mathrm{loop}_1,\cdots, \mathrm{loop}_k$ with no crossings between themselves or each other. Then, one defines $q_{\eta}: H_1(\Sigma,\Z_2) \to \Z_2$.
\begin{equation} \label{eq:quad_form_def}
    (-1)^{q_{\eta}(\xi)} = \prod_{\mathrm{loops}} (-1)^{1 + \mathrm{wind(loop)}} = (-1)^{\# \, \mathrm{of \, loops}} \prod_{\mathrm{loops}} (-1)^{\mathrm{wind(loop)}}.
\end{equation}
A priori, this function isn't well-defined because it depends on the specific representatives of $\xi$. The theorem below shows it is well-defined and a crucial `quadratic' property. See Fig.~\ref{fig:spin_structs_torus} for a depiction of all spin structures, related vector fields, and $q_\eta(\xi)$ for each homology class.

\begin{thm}[\cite{johnson1980}]
    The Eq.~\eqref{eq:quad_form_def} defined for a collection of loops with no intersections is independent of the choice of loops used to represent $\xi$. Moreover, $q_\eta(\xi)$ satisfies the \textit{quadratic refinement property}
    \begin{equation}
        (-1)^{q_{\eta}(\xi + \xi')} = (-1)^{q_{\eta}(\xi)} (-1)^{q_{\eta}(\xi')} (-1)^{\mathrm{int}(\xi, \xi')},
    \end{equation}
    where $\mathrm{int}(\xi, \xi')$ is the intersection pairing of $\xi,\xi'$, i.e. the (mod 2) number of generic intersections between $\xi, \xi'$. 
\end{thm}
\begin{proof}[Sketch of Proof]
    First for grounding, we note that $(-1)^{1 + \mathrm{wind(loop)}}$ is a natural factor, since for a small trivial loop we have $\mathrm{wind(loop)} = \pm 1$, $(-1)^{1 + \mathrm{wind(loop)}} = 1$ so that adding or removing trivial loops that don't change the homology class don't change the quantity.
    
    To show the quantity is well-defined, it suffices to show that any two collections of non-intersecting loops acting representatives of a $\Z_2$-homology class $\xi$ can be obtained from each other by regular homotopy, adding collections of trivial loops in general position with respect to the preexisting loops, resolving intersections, and removing trivial loops. This can be shown adding the difference of any two chain representatives of $\xi$, taking care to possibly deform loops slightly so that one can resolve intersections and remove trivial loops. Crucially, adding trivial loops will create an even number of intersections, which changes the total number of loops by an even number, which leaves $(-1)^{\# \, \mathrm{of \, loops}}$ invariant.

    The quadratic refinement property is essentially Lemma~\ref{lem:resolving_intersection}. In particular, consider collections of non-intersecting loops $\mathcal{C}_{\xi}$ representing $\xi$ and $\mathcal{C}_{\xi'}$ representing $\xi'$, so that each collection is in general position with respect to each other. Then the collection $\mathcal{C}_{\xi} \cup \mathcal{C}_{\xi'} $ will be a representative of $\xi + \xi'$ having intersections only between $\mathcal{C}_{\xi}$ and $\mathcal{C}_{\xi'}$. The number of such intersections (mod 2) agrees with $\mathrm{int}(\xi, \xi')$ Resolving these intersections gives a new collection of loops $\mathcal{C}_{\xi + \xi'}$
    
    The two statements in Lemma~\ref{lem:resolving_intersection} imply the two equalities
    \begin{equation} \label{eq:johnson_thm_eq1}
        |\mathcal{C}_{\xi}| + |\mathcal{C}_{\xi'}| = |\mathcal{C}_{\xi+\xi'}| + \mathrm{int}(\xi,\xi') \,\, \mathrm{(mod \, 2)}
        \quad\quad\quad\text{and}\quad\quad\quad
        \prod_{\mathrm{loop} \in \mathcal{C}_{\xi} \cup \mathcal{C}_{\xi'}}
        (-1)^{\mathrm{wind(loop)}}
        =
        \prod_{\mathrm{loop} \in \mathcal{C}_{\xi+\xi'}}
        (-1)^{\mathrm{wind(loop)}}
    \end{equation}
    As such, we'll have
    \begin{equation}
    \begin{split}
        (-1)^{q_{\eta}(\xi)} \times (-1)^{q_{\eta}(\xi')} 
        &=
        (-1)^{|\mathcal{C}_{\xi}|} 
        \prod_{\mathrm{loop} \in \mathcal{C}_{\xi}} (-1)^{\mathrm{wind(loop)}}
        \times
        (-1)^{|\mathcal{C}_{\xi'}|} 
        \prod_{\mathrm{loop} \in \mathcal{C}_{\xi'}} (-1)^{\mathrm{wind(loop)}}
        \\
        &=
        (-1)^{|\mathcal{C}_{\xi}| + |\mathcal{C}_{\xi'}|} 
        \prod_{\mathrm{loop} \in \mathcal{C}_{\xi} \cup \mathcal{C}_{\xi'}}
        (-1)^{\mathrm{wind(loop)}}
        \\
        &=
        (-1)^{|\mathcal{C}_{\xi+\xi'}| + \mathrm{int}(\xi, \xi')} 
        \prod_{\mathrm{loop} \in \mathcal{C}_{\xi + \xi'}}
        (-1)^{\mathrm{wind(loop)}}
        =
        (-1)^{\mathrm{int}(\xi, \xi')} (-1)^{q_{\eta}(\xi+\xi')}.
    \end{split}
    \end{equation}
    The first and fourth equalities use the definitions of $q_\eta(\xi),q_\eta(\xi'),q_\eta(\xi+\xi')$. The third equality uses Eq.~\eqref{eq:johnson_thm_eq1}. 
\end{proof}

From here, we will study the quadratic form $q_\eta(\xi)$. It will be useful to consider the square matrix $A_{\eta,\xi} = (-1)^{q_\eta(\xi)}$, which is $2^{2g} \times 2^{2g}$ matrix. 
\begin{lem} \label{lem:inv_quad_form}
    $A_{\eta,\xi} = (-1)^{q_\eta(\xi)}$ has orthogonal rows, and its inverse is 
    $(A^{-1})_{\xi,\eta} = \frac{1}{2^{2g}} (-1)^{q_\eta(\xi)}$
\end{lem}
\begin{proof}
    Consider two rows corresponding to $\eta \neq \eta'$. Then, we have that for some $0 \neq a \in H^1(\Sigma,\Z_2)$, $\eta' = a \cdot \eta$. So, we have:
    \begin{equation}
        A_{\eta',\xi} = (-1)^{q_{\eta'}(\xi)} = (-1)^{q_{\eta}(\xi)} (-1)^{\braket{a,\xi}}
        =
        (-1)^{\braket{a,\xi}} A_{\eta,\xi}.
    \end{equation}
    Now since $a \neq 0$, we have that $\sum_{\xi \in H_1(\Sigma,\Z_2)} (-1)^{\braket{a,\xi}} = 0$ 
    which implies
    \begin{equation}
        \sum_{\xi \in H_1(\Sigma,\Z_2)} A_{\eta,\xi} A_{\eta',\xi}
        =
        \sum_{\xi \in H_1(\Sigma,\Z_2)} (-1)^{\braket{a,\xi}} = 0.
    \end{equation}

    Since $A$ has orthogonal rows and each row has a squared-length ${\sum_{\xi \in H_1(\Sigma,\Z_2)} |(-1)^{q_{\eta(\xi)}}|^2} = 2^{2g}$, its inverse is exactly
    $(A^{-1})_{\xi,\eta} = \frac{1}{2^{2g}} (-1)^{q_\eta(\xi)}$.
\end{proof}

We can define the so-called `Arf Invariant' of a spin structure $\eta$ by scaling the column sums of $A$. We define
\begin{equation}
    \mathrm{Arf}(\eta)
    =
    \frac{1}{2^g} 
    \sum_{\xi \in H_1(\Sigma,\Z_2)} 
    (-1)^{q_\eta(\xi)}.
\end{equation}
We have the following fact that $\mathrm{Arf}$ is always $\pm 1$
\begin{lem}
    $\mathrm{Arf}(\eta) = \pm 1 \quad \text{for all} \,\, \eta$
\end{lem}
\begin{proof}
    Consider $(\mathrm{Arf}(\eta))^2$. We have
    \begin{equation}
    \begin{split}
        (\mathrm{Arf}(\eta))^2
        &= 
        \left(
            \frac{1}{2^g} 
            \sum_{\xi \in H_1(\Sigma,\Z_2)} 
            (-1)^{q_\eta(\xi)}
        \right)
        \left(
            \frac{1}{2^g} 
            \sum_{\xi' \in H_1(\Sigma,\Z_2)} 
            (-1)^{q_\eta(\xi')}
        \right)
        \\
        &=
        \frac{1}{2^{2g}} 
        \sum_{\xi , \xi' \in H_1(\Sigma,\Z_2)} 
        (-1)^{q_\eta(\xi)} (-1)^{q_\eta(\xi')} 
        =
        \frac{1}{2^{2g}} 
        \sum_{\xi , \xi' \in H_1(\Sigma,\Z_2)} 
        (-1)^{q_\eta(\xi+\xi')} (-1)^{\mathrm{int}(\xi,\xi')}.
    \end{split}
    \end{equation}
    Now, let's change variables in the sum $\Tilde{\xi} = \xi + \xi'$. Note that by linearity of the intersection pairing and the fact that $\mathrm{int}(\xi,\xi) = 0$ for all $\xi$, we have $\mathrm{int}(\xi,\Tilde{\xi}) = \mathrm{int}(\xi,\xi')$. We we can rewrite the above
    \begin{equation}
    \begin{split}
        (\mathrm{Arf}(\eta))^2
        &=
        \frac{1}{2^{2g}} 
        \sum_{\xi , \Tilde{\xi} \in H_1(\Sigma,\Z_2)} 
        (-1)^{q_\eta(\Tilde{\xi})} (-1)^{\mathrm{int}(\xi,\Tilde{\xi})}.
    \end{split}
    \end{equation}
    Now, we can use the facts that $\sum_{\xi} (-1)^{\mathrm{int}(\xi,\Tilde{\xi})} = 2^{2g} \delta_{\Tilde{\xi},0}$ and $(-1)^{q_\eta(0)} = 1$ for all $\eta$ to give 
    \begin{equation}
    \begin{split}
        (\mathrm{Arf}(\eta))^2
        =
        \frac{1}{2^{2g}} 
        \sum_{\Tilde{\xi} \in H_1(\Sigma,\Z_2)} 
        (-1)^{q_\eta(\Tilde{\xi})} 2^{2g} \delta_{\Tilde{\xi},0}
        =
        1.
    \end{split}
    \end{equation}
    This implies $\mathrm{Arf}(\eta) = \pm 1$.
\end{proof}

We note that the Arf invariant can be used to distinguish spin structures. For example, on the torus there are four inequivalent spin structures which we can call $P/P, AP/P, P/AP, AP/AP$ based on if it induces a periodic or antiperiodic spin structure around the cycle basis. However, acting on the torus via the mapping class group like a Dehn twist will relabel the basis of cycles and change the spin structure. Indeed, one can transition between the $AP/P$, $P/AP$, $AP/AP$ spin structures by Dehn twists, so in a sense those spin structures are isomorphic under the mapping class group even though they're not the same. However, the $P/P$ never gets changed under the mapping class group. This is reflected in the Arf invariant because the Arf invariant is manifestly invariant under the mapping class group; it will just rearrange the labeling of the sum arguments $\xi \in H_1(\Sigma,\Z_2)$.

\bibliography{bibliography}
        
\end{document}